\documentclass[preprint]{elsarticle}
\biboptions{sort&compress}
\usepackage{tikz}
\usetikzlibrary{fit,positioning}
\usetikzlibrary{arrows}
\usetikzlibrary{automata}
\usetikzlibrary{calc}
\usetikzlibrary{decorations.pathmorphing}
\usetikzlibrary{decorations.pathreplacing}
\tikzset{every state/.style={minimum size=0pt}}
\tikzset{
    position/.style args={#1:#2 from #3}{
        at=(#3.#1), anchor=#1+180, shift=(#1:#2)
    }
}

\usepackage{amsmath}
\usepackage{amssymb} % more Symbols
\usepackage{amsfonts}
\usepackage{theorem}
\usepackage[basic]{complexity}
\usepackage{shuffle}
\usepackage{bm}
\usepackage{url}
\usepackage{mathabx}
\usepackage{mathtools}%for \shortintertext
\usepackage{booktabs}

\newproof{claim}{Claim}%number 1,2,3 .. independently of other theorems
\newtheorem{theorem}{Theorem}[section]

\newtheorem{lemma}[theorem]{Lemma}
\newtheorem{proposition}[theorem]{Proposition}
\newtheorem{remark}[theorem]{Remark}
\newproof{proof}{Proof}

\newcommand{\Nat}{{\mathbb{N}}}
\renewcommand{\setminus}{\smallsetminus}

\newcommand{\size}[1]{\mathopen{\mid}#1\mathclose{\mid}}
\newcommand{\step}[1]{\xrightarrow{\!\!#1\!\!}}
\newcommand{\stepR}[1]{\step{#1}_R}
\newcommand{\eqdef}                     % equal by definition
  {\stackrel{\scriptscriptstyle \mathrm{def}}{=}}
\newcommand{\equivdef}                  % equivalent by definition
  {\stackrel{\scriptscriptstyle \mathrm{def}}{\iff}}
\newcommand{\embeds}{\sqsubseteq}
\newcommand{\subword}{\embeds}% synonymous

\newcommand{\upint}{{{\downtouparrow}}}
\newcommand{\downint}{{{\uptodownarrow}}}
\newcommand{\down}{{{\downarrow}}}
\newcommand{\up}{{{\uparrow}}}
\newcommand{\nDFA}{{n_{\textrm{D}}}}
\newcommand{\nUFA}{{n_{\textrm{U}}}}
\newcommand{\nNFA}{{n_{\textrm{N}}}}
\newcommand{\nAFA}{{n_{\textrm{A}}}}
\newcommand{\nNDFA}{{n_{\mathrm{N\&D}}}}

\newcommand{\init}{{\textrm{init}}}
\newcommand{\rank}{{\textit{rank}}}
\newcommand{\emptyword}{\varepsilon}
\makeatletter
\DeclareRobustCommand{\rvdots}{%
  \vbox{
    \baselineskip4\p@\lineskiplimit\z@
    \kern-\p@
    \hbox{.}\hbox{.}\hbox{.}
  }}
\makeatother

\begin{document}

\title{On the state complexity of closures and interiors of regular languages with
subwords and superwords}
\author[cmi,lsv]{P.~Karandikar \fnref{tcs}}
\author[bay]{M. Niewerth \fnref{dfg}}
\author[lsv]{Ph. Schnoebelen \fnref{anr}}

\address[cmi]{Chennai Mathematical Institute}
\address[bay]{University of Bayreuth}
\address[lsv]{LSV, ENS Cachan, CNRS}

\fntext[tcs]{Partially funded by Tata Consultancy Services.}
\fntext[dfg]{Supported by Grant MA 4938/21 of the DFG.}
\fntext[anr]{Supported by Grant ANR-11-BS02-001.}

\begin{abstract}
The downward and upward closures of a regular language $L$
are obtained by collecting all the subwords and
superwords of its elements, respectively. The downward and upward interiors of $L$ are obtained dually by
collecting words having all their subwords and superwords in $L$, respectively.
We provide lower and upper bounds on the size of the smallest automata
recognizing these closures and interiors. We also consider the
computational complexity of decision problems for closures of regular
languages.
\end{abstract}

%% Local Variables:
%% ispell-check-comments: nil
%% ispell-local-dictionary: "english"
%% fill-column: 75
%% End:

% LocalWords:  superwords


\maketitle

%% Local Variables:
%% ispell-check-comments: nil
%% ispell-local-dictionary: "english"
%% fill-column: 75
%% End:

% LocalWords:

\section{Introduction}
%=====================
\label{sec-intro}

State complexity is a standard measure of the descriptive complexity of
regular languages. The most common state complexity problems ask, given a
regularity-preserving operation $\operatorname{f}$ on languages, to bound the size of
an automaton recognizing $\operatorname{f}(L)$ when $L$ is recognized by an
$n$-state automaton. We refer to~\cite{holzer2003b,yu2005} for a survey of the main
known results in the area.

In this article, we consider language operations based on subwords.
Recall that a (scattered) subword of some word $x$ is a word obtained from $x$ by
removing any number of letters at arbitrary positions in $x$, see formal definitions in Section~\ref{sec-basics}. Symmetrically, a superword is
obtained by inserting letters at arbitrary positions. Subwords and
superwords occur in many areas of computer science, from searching in texts
and databases~\cite{baezayates91} to the theory of codes~\cite{ito2000},
computational linguistics~\cite{paun97}, and DNA computing~\cite{kari99}.

For a language $L\subseteq \Sigma^*$, we write $\down L$ for the set of all
its subwords and $\up L$ for the set of all its superwords (in $\Sigma^*$)
and call them the \emph{downward closure} and \emph{upward closure} of $L$, respectively. Dual
to closures are \emph{interiors}. The \emph{upward interior} and
\emph{downward interior} of $L$, denoted $\upint L$ and $\downint L$, are
the largest upward-closed and downward-closed sets included in
$L$.
It
has been known since~\cite{haines69} that $\down L$ and $\up L$ are regular
for any $L$. Then $\downint L$ and $\upint L$ are regular too by duality, as
expressed in the following equalities:
\begin{xalignat}{2}
\label{eq-interiors-are-duals}
\downint L&=\Sigma^*\setminus \up(\Sigma^*\setminus L)
\:,
&
\upint L&=\Sigma^*\setminus \down(\Sigma^*\setminus L)
\:.
\end{xalignat}

Computing closures and interiors has several applications in
computer-aided reasoning~\cite{KS-fosubw} and program verification.
Computing closures is an essential ingredient in
the verification of safety properties of channel systems
-- see~\cite{abdulla-forward-lcs,hss-lmcs} -- while computing interiors
is required for the verification of their game-theoretical
properties~\cite{BS-fmsd2013}. More generally, the
regularity of upward and downward closures make them good
overapproximations of more complex languages -- see~\cite{habermehl2010,bachmeier2015,zetzsche2015} -- and interiors can be
used as regular underapproximations.

Recently Gruber \textit{et al}.\  explicitly raised the issue
of the state complexity of downward and upward closures of regular
languages~\cite{gruber2007,gruber2009} (less explicit precursors exist,
for example,~\cite{birget93}). Given an $n$-state automaton $A$ that recognizes $L$,
 automata $A^\down$ and $A^\up$ that recognize $\down L$ and,
$\up L$ respectively can be obtained by simply adding extra transitions to
$A$. However, when $A$ is a deterministic automaton (a DFA), the resulting
$A^\down$ and $A^\up$ are in general not deterministic (are NFAs), and their
determinization may entail an exponential blowup. With $n$
denoting the number of states of $A$, Gruber \textit{et al}.\  proved a
$2^{\Omega(\sqrt{n}\log n)}$ lower bound on the number of states of any DFA
recognizing $\down L$ or $\up L$~\cite{gruber2009}, to be compared with the $2^n-1$ upper bound that
comes from the simple closure+determinization method.

Okhotin improved on these results by showing an improved
$2^{\frac{n}{2}-2}$ lower bound for $\down L$. He also established the
exact state complexity for $\up L$ by proving a $2^{n-2}+1$ upper bound
and showing that this is tight~\cite{okhotin2010}.

All the above lower bounds assume an unbounded alphabet, and Okhotin showed
that his $2^{n-2}+1$ state complexity for $\up L$ requires $n-2$ distinct
letters. He then considered the case of languages over a \emph{fixed
  alphabet} and, in the 3-letter case, he demonstrated
exponential $2^{\sqrt{2n+30}-6}$ and $\frac{1}{5} 4^{\sqrt{n/2}}
n^{-\frac{3}{4}}$ lower bounds for $\down L$ and $\up
L$ respectively~\cite{okhotin2010}. In the $2$-letter case, H\'eam had previously proved
an $\Omega(r^{\sqrt{n}})$ lower bound for $\up L$, here with
$r=(\frac{1+\sqrt 5}{2})^{\frac{1}{\sqrt 2}}$~\cite{heam2002}. Regarding
$\down L$, the question whether its state complexity is
exponential even when $\size{\Sigma}=2$ was left open (note that the one-letter case
is trivial).\\

The state complexity of interiors has not yet been considered in the
literature. When working with DFAs, complementation is essentially free so
that computing interiors reduces to computing closures, thanks to duality.
However, when working with NFAs, the simple complement+closure+complement
method comes with a quite large $2^{2^n}$ upper-bound on the number of
states of an NFA that recognizes $\upint L$ or $\downint L$ -- it actually yields
DFAs -- and one would like to improve on this, or to prove a matching lower
bound. As we explain in Section~\ref{ssec-afas}, this is related to the
state complexity of closures when working with alternating automata (AFAs),
a question recently raised in~\cite{holub2014}. \\

\noindent \emph{Our contribution.} Regarding closures with DFAs, we prove
in Section~\ref{sec-closure} a tight $2^{n-1}$ state complexity for
downward closure and show that its tightness requires unbounded alphabets.
In Section~\ref{sec-2-letter} we prove an exponential lower bound on both
$\down L$ and $\up L$ in the case of a two-letter alphabet, answering the
open question raised above.
Regarding interiors on NFAs, we show in Section~\ref{sec-interior}
doubly-exponential lower bounds for downward and upward interiors, assuming
an unbounded alphabet. We also provide improved upper bounds, lower than
the naive $2^{2^n}$ but still doubly exponential. Table~\ref{tab-summary}
shows a summary of the results.
Finally, Section~\ref{sec-ufa} proves lower bounds on unambiguous
automata for the witness languages used in Section~\ref{sec-closure}, and Section~\ref{sec-decision-pbs} considers the computational
complexity of some basic decision problems for sets of subwords or superwords
described by automata.
\\

% CAVEAT (phs/7 Sep 2015): This table uses several formatting tricks to
% reduce cluttering and width. Relation symbols like =, \geq and \leq are
% inside {..} so that they stick to their arguments. There are also some
% \!'s that artificially shrink the spacing. PLEASE DO NOT FIX WITHOUT
% CONSIDERING THE LOOK AND FEEL OF THE OUTCOME.
\begin{table}[htb]
\centering
\caption[caption]{A summary of the results on state complexity for closures and interiors, where $\psi(n)$ (${\leq} 2^{2^n}$) is the $n$th Dedekind's number{\protect\footnotemark}.}
\begin{tabular}{@{}llc@{}}
\toprule
\multicolumn{1}{c}{Operation} & \multicolumn{1}{c}{Unbounded alphabet}  & \multicolumn{1}{c}{Fixed alphabet} \\
\midrule
$\up L$ (DFA to DFA) & ${=}2^{n{-}2}+1$ for $\size{\Sigma}{\geq} n{-}2$ & $2^{\Omega(n^{1/2})}$ for $\size{\Sigma}{=}2$ \\
$\down L$ (DFA to DFA) & ${=}2^{n{-}1}$ for $\size{\Sigma}{\geq} n{-}1$ & $2^{\Omega(n^{1/3})}$ for $\size{\Sigma}{=}2$ \\
\midrule
$\up L$ (AFA to AFA) & ${\geq} 2^{\left\lfloor\! \frac{n{-}3}{2}\!\right\rfloor}$ and ${<} 2^n$ for $\size{\Sigma}$ in $2^{\Omega(n)}$ & $\vdots$ \\
$\down L$ (AFA to AFA) & ${\geq} 2^{\left\lfloor\! \frac{n{-}4}{3}\!\right\rfloor}$ and ${\leq} 2^n$ for $\size{\Sigma}$ in $2^{\Omega(n)}$ & \textrm{(unknown)} \\
\midrule
$\upint L$ (NFA to NFA) & ${>} 2^{2^{\left\lfloor\! \frac{n{-}4}{3}\!\right\rfloor}}$ and ${\leq}\psi(n)$ for $\size{\Sigma}$ in $2^{\Omega(n)}\!\!\!\!\!\!$ & $\vdots$ \\
% use > rather than \geq with +1 as in Prop~\ref{prop-nfa-upint}
$\downint L$ (NFA to NFA) & ${\geq} 2^{2^{\left\lfloor\! \frac{n{-}3}{2}\!\right\rfloor}}$ and ${\leq}\psi(n)$ for $\size{\Sigma}$ in $2^{\Omega(n)}\!\!\!\!\!\!$ & $\vdots$ \\
\bottomrule
\end{tabular}
\label{tab-summary}
\end{table}
\footnotetext{Recall that the $n$th Dedekind number $\psi(n)$ is the number of
antichains in the lattice of subsets of an $n$-element set, ordered by
inclusion~\cite{kleitman69}. Kahn~\cite[Corollary~1.4]{kahn2002} shows
\[
\binom{n}{\lfloor n/2\rfloor}
\leq \log_2 \psi(n)
\leq
\left(1+\frac{2 \log(n+1)}{n}\right) \binom{n}{\lfloor n/2\rfloor}
\:.
\]
}
%% Local Variables:
%% fill-column: 999
%% End:

\noindent \emph{Related work.}
We already mentioned previous work on the closure of regular languages:
it is also possible to compute closures by subwords or superwords for
larger classes like context-free languages or Petri net languages,
see~\cite{habermehl2010,bachmeier2015,zetzsche2015} and the references
therein for applications and some results on descriptive complexity.

Interiors are duals of closures and should not be confused with the inverse
operations considered in~\cite{bianchi2012b}, or the shuffle residuals from
\cite{ito2000}. Duals of regularity-preserving operations have the form
``\textit{complement--operation--complement}'' and thus can be seen as
special cases of the \emph{combined operations} studied
in~\cite{salomaa2007} and following papers.  Dual operations occur
naturally in algorithmic or logical contexts but have not yet been
considered widely from a state-complexity perspective: we are only aware
of~\cite{birget96} studying the dual of $L\mapsto \Sigma^*\cdot L$.

%% Local Variables:
%% ispell-check-comments: nil
%% ispell-local-dictionary: "english"
%% fill-column: 75
%% TeX-master: "main"
%% End:

% LocalWords:  Gruber et al Okhotin DFAs NFAs closedness eam wrt AFAs

\section{Basic notions and results}
%==================================
\label{sec-basics}

\paragraph*{Subwords}

We assume familiarity with regular languages and the automata that
recognize them. We write $x,y,u,v,\ldots$ to denote words over a finite
alphabet $\Sigma=\{a,b,\ldots\}$, with $\size{x}$ denoting the
length of a word $x$. For $1\leq i\leq|x|$,  we let $x[i]$ denote the
$i$-th letter of $x$. The empty word is denoted
$\emptyword$ and concatenation is denoted multiplicatively.

We say that a word $x$ is a \emph{subword} of $y$, written $x\subword y$,
when $y$ can be written in the form $y=y_0\,x_1\,y_1\cdots
y_{m-1}\,x_m\,y_m$ for some factors such that $x=x_1\cdots x_m$.
For example, $\emptyword\subword a\,b\subword a\,c\,b\,a$.
Equivalently, $x\subword y$ when
there are positions $0 < p_1 < p_2 < \cdots < p_\ell \leq \size{y}$ such that
$x[i]=y[p_i]$ for all $1\leq i\leq \ell=\size{x}$.
When $x\subword y$ we also say that $x$ \emph{embeds} in $y$, or that $y$ is a
\emph{superword} of $x$.

\paragraph*{Closures}

For a language $L\subseteq\Sigma^*$, we define $\up L
\eqdef\{x\in\Sigma^*~|~\exists y\in L:y\subword x\}$ and $\down
L\eqdef\{x\in\Sigma^*~|~\exists y\in L:x\subword y\}$, and call them the upward
and downward closures of $L$ respectively.\footnote{Formally $\up L$ should more
  precisely be denoted $\up_\Sigma L$ since it depends on the underlying
  alphabet but in the rest of this article $\Sigma$ will always be clear from
  the context.}

The Kuratowski closure axioms are satisfied:
\[
\down \emptyset=\emptyset
\:,
\quad
L\subseteq\down L=\down\down L
\:,
\quad
\down\bigl(\bigcup_i L_i\bigr)=\bigcup_i\down L_i
\:,
\quad
\down \bigl(\bigcap_i \down L_i\bigr)=\bigcap_i\down L_i
\:,
\]
and similarly for upward closures.
We say that a language $L\subseteq\Sigma^*$ is \emph{downward-closed}
if $L=\down L$ and that a language is \emph{upward-closed} if $L=\up L$. Note
that $L$ is downward-closed if, and only if, its complement $\Sigma^*\setminus L$ is
upward-closed but the complement of $\down L$ is \emph{not} $\up L$.

\paragraph*{Regularity}

The upward-closure $\up x$ of a word $x=a_1\cdots a_\ell$ is a regular
language given by the regular expression $\Sigma^* a_1 \Sigma^* \cdots
a_\ell\Sigma^*$. Since, by Higman's Lemma~\cite{higman52}, any language $L$ only contains finitely many
 elements that are minimal for the subword ordering, one deduces that $\up{L}$ is
regular for any $L\subseteq\Sigma^*$, a result also known as Haines's Theorem~\cite{haines69}. Then $\down{L}$, being the complement
of an upward-closed language, is regular too. In fact, upward-closed
languages are simple star-free languages. They correspond exactly to the
level $\frac{1}{2}$ of Straubing's hierarchy~\cite{pin97}, and coincide with
the \emph{shuffle ideals}, that is, the languages that satisfy
$L=L\shuffle \Sigma^*$~\cite{brzozowski2013}. Downward-closed languages coincide with
\emph{strictly piecewise-testable} languages~\cite{rogers2010}.

Effective construction of a finite-state automaton recognizing $\down L$ or $\up L$
is easy when $L$ is regular (see Section~\ref{sec-closure}), is possible
when $L$ is context-free~\cite{leeuwen78,courcelle91}, and is not
possible in general since this would allow deciding the emptiness of $L$.

\paragraph*{Interiors}

The \emph{upward interior} of a language $L$ over $\Sigma$ is $\upint
L\eqdef\{x\in\Sigma^*~|~\up x\subseteq L\}$. Its \emph{downward interior}
is $\downint L\eqdef\{x\in\Sigma^*~|~\down x\subseteq L\}$. Alternative
characterizations are possible, for example, by noting that $\upint L$ or $\downint L$ are the largest upward-closed or downward closed languages, respectively, included in $L$, or by using
the duality equations from page~\pageref{eq-interiors-are-duals}.
These equations show that $\downint L$ and $\upint
L$ are regular for any $L$. They also show how, when $L$ is regular, one may
compute automata recognizing the interiors of $L$ by combining complementations and
closures.

\paragraph*{State complexity}
%----------------------------
When considering a finite automaton $A=(\Sigma,Q,\delta,I,F)$, we usually
write $n$ for $\size{Q}$, $k$ for $\size{\Sigma}$, and $L(A)$ for the language recognized by $A$.
In the context of a fixed automaton $A$ we often write $q\step{a}q'$ to
mean $q'\in\delta(q,a)$.
We also write $q\step{w}q'$ where $w\in\Sigma^*$ to denote the existence of a $w$-labeled path from $q$ to $q'$
 in the graph of $A$.
In Section~\ref{sec-ufa} we consider unambiguous
automata (UFAs): recall that an NFA $A$ is \emph{unambiguous} if every word $w\in
L(A)$ has a single accepting run~\cite{colcombet2015}.

For a regular language $L$, $\nDFA(L)$, $\nNFA(L)$ and $\nUFA(L)$ denote the
minimum number of states of a DFA, an NFA, and a UFA, respectively, that accepts $L$. Note that NFAs are allowed to have multiple initial states, and DFAs need not be complete.
Since any DFA is unambiguous, one obviously has $\nNFA(L)\leq \nUFA(L)\leq
\nDFA(L)$ for any regular language. In cases where $\nNFA(L)=\nDFA(L)$ we
may use $\nNDFA(L)$ to denote the common value.

\paragraph*{An application of the fooling set technique}

The following lemma is a well-known tool for proving lower bounds on $\nNFA(L)$.

\begin{lemma}[Extended fooling set technique, \cite{gruber2006}]
\label{extfool}
Let $L$ be a regular language. Suppose that there exists a set of pairs of words
$S = \{(x_i, y_i)\}_{1 \leq i \leq n}$, called a \emph{fooling set}, such that $x_i\,y_i
\in L$ for all $i=1,\ldots,n$, and such that for all $j \neq i$, at least one of $x_i \, y_j$ and $x_j \, y_i$ is not in $L$.
Then $\nNFA(L)\geq n$.
\end{lemma}
\begin{proof}
Let $A=(\Sigma,Q,\delta, I,F)$ be an NFA recognizing $L$. For each $i=1,\ldots,n$,
$x_i y_i \in L$, so $A$ has an accepting run of the form $s_i \step{x_i}
q_i \step{y_i} f_i$, starting at some initial state $s_i\in I$, ending at
some accepting state $f_i\in F$, and visiting some intermediary state
$q_i\in Q$. Observe that if $q_i = q_j$ for $i \neq j$ then $A$ has
accepting runs for both $x_i y_j$ and $x_j y_i$, which contradicts the
assumption. Hence the states $q_1, q_2, \ldots, q_n$ are all distinct and
$\size{Q}\geq n$.
\qed
\end{proof}

In preparation for Section~\ref{sec-closure}, let us apply the fooling set technique to the following languages, where $\Sigma$ is an
arbitrary finite alphabet:
\begin{xalignat*}{2}
U_\Sigma &\eqdef \{x\in\Sigma^*~|~\forall a\in\Sigma:\exists i:x[i]=a\} \:,
&
U_\Sigma' &\eqdef \Sigma \cdot U_\Sigma \:,
\\
V_\Sigma &\eqdef \{x\in\Sigma^*~|~\forall i\neq j:x[i]\neq x[j]\} \:.
\end{xalignat*}
  Note that $U_\Sigma$ consists of all words where every letter in $\Sigma$ appears
  at least once while $V_\Sigma$ consists of all words where no letter appears
  twice. A word in $U_\Sigma'$ consists of an arbitrary letter from $\Sigma$ followed by a word in
  $U_\Sigma$. Note that $U_\Sigma$ and $U'_\Sigma$ are upward-closed while
  $V_\Sigma$ is downward-closed.
\begin{lemma}
\label{lem-U-n-V}
$\nNDFA(U_\Sigma)=\nNDFA(V_\Sigma)=2^{\size{\Sigma}}$ and, if $\Sigma$ is
not empty,
$\nNDFA(U_\Sigma')= 2^{\size{\Sigma}}+1$.
\end{lemma}
\begin{proof}
  Let us start with the lower bounds for $\nNFA(U_\Sigma)$ and
  $\nNFA(V_\Sigma)$: With any $\Gamma\subseteq\Sigma$, we associate two
  words $x_\Gamma$ and $x_{\neg\Gamma}$, where $x_\Gamma$  has exactly one
  occurrence of each letter from $\Gamma$, and where $x_{\neg\Gamma}$ has
  exactly one occurrence of each letter not in $\Gamma$. Then $x_\Gamma x_{\neg\Gamma}$
  belongs to $U_\Sigma$ and $V_\Sigma$, while for any $\Delta\neq\Gamma$
  one of $x_\Gamma x_{\neg\Delta}$ and $x_\Delta x_{\neg\Gamma}$ does not belong to
  $U_\Sigma$ and one does not belong to $V_\Sigma$. Thus for $U_\Sigma$
  or $V_\Sigma$ we may use the same fooling set
  $S=\{(x_\Gamma,x_{\neg\Gamma})\}_{\Gamma\subseteq\Sigma}$. By Lemma~\ref{extfool}, we conclude that
  $\nNFA(U_\Sigma)\geq2^{|\Sigma|}$ and $\nNFA(V_\Sigma)\geq2^{|\Sigma|}$.

  For $U_\Sigma'$ we pick an arbitrary letter $a\in\Sigma$ and let our
  fooling set be
  $S=\{(a\,x_\Gamma,x_{\neg\Gamma})\}_{\Gamma\subseteq\Sigma}\cup\{(\emptyword,a\,x_\Sigma)\}$.
  As above $a\,x_\Gamma x_{\neg\Gamma}$ belongs to $U_\Sigma'$ while, for any
  $\Delta\neq\Gamma$, one of $a\,x_\Gamma x_{\neg\Delta}$ and $a\,x_\Delta
  x_{\neg\Gamma}$ does not belong to $U_\Sigma'$. Furthermore $\emptyword\cdot
  a\,x_\Sigma$ belongs to $U_\Sigma'$, while $\emptyword\cdot x_{\neg\Gamma}$
  does not belong to $U_\Sigma'$ for any $\Gamma\subseteq\Sigma$. By
  Lemma~\ref{extfool}, we conclude that $\nNFA(U_\Sigma') \geq 2^{\size{\Sigma}}+1$.

  Proving the upper bounds is a well-known exercise in automata theory.
  One
  designs DFAs using the powerset $2^\Sigma=\{\Gamma,\Gamma',\ldots\}$ as the set of states, that is,
  automata with $2^{\size{\Sigma}}$ states. With rules of the form
  $\Gamma\step{a}\Gamma\cup\{a\}$, these states
  record the set of letters read so far, starting from
   $\emptyset$ as initial state. In the DFA for $U_\Sigma$,
  one accepts when all letters have been seen. In the DFA for $V_\Sigma$,
  all states are accepting but it is forbidden to read a letter that has
  already been seen: there are no transitions
  $\Gamma\step{a}\Gamma\cup\{a\}$ when $a\in\Gamma$. A DFA recognizing
  $U'_\Sigma$ is obtained from the DFA for $U_\Sigma$ by adding a
  new initial state from which one will read a first
  letter before continuing as for $U_\Sigma$.
  \qed
\end{proof}
In the following, we use $\Sigma_k\eqdef\{a_1,\ldots,a_k\}$ to denote a
$k$-letter alphabet, and  write $U_k$ and $V_k$ instead of $U_{\Sigma_k}$
and $V_{\Sigma_k}$.
  
%% Local Variables:
%% ispell-check-comments: nil
%% ispell-local-dictionary: "english"
%% fill-column: 75
%% TeX-master: "main"
%% End:

% LocalWords:  CoNRS wrt Eq UFAs th Kuratowski

\section{State complexity of closures}
%=====================================
\label{sec-closure}

Let $L\subseteq\Sigma^*$ be a regular language recognized by an NFA $A$.
One may obtain NFAs recognizing the upward and downward closures $\up L$ and $\down
L$ by simply adding transitions to $A$, without increasing its number of
states. More precisely, an NFA $A^\up$ recognizing $\up L$ is obtained from $A$ by
adding self-loops $q \step{a} q$ for every state $q$ of $A$ and every
letter $a\in\Sigma$. Similarly, an NFA $A^\down$ recognizing $\down L$ is obtained
from $A$ by adding a silent transition $p \step{\emptyword} q$, also called
an ``$\emptyword$-transition'', for every original transition $p \step{a} q$
in $A$.

If $L$ is recognized by a DFA or an NFA $A$ and we want a DFA recognizing $\up L$
or $\down L$, we can start with the NFA $A^\up$ or $A^\down$ defined
above and transform it into a DFA using the powerset construction. This
shows that if $L$ is recognized by an $n$-state DFA, then both its upward
and downward closures are recognized by DFAs with at most $2^n-1$ states.

It is possible to provide tighter upper bounds by taking
advantage of specific features of $A^\up$ and $A^\down$. The next two
propositions give tight upper bounds for upward and downward closure, respectively.

\begin{proposition}[State complexity of upward closure, after~\cite{okhotin2010}]
\label{prop-sc-upward-dfa}
1.\  If $L\subseteq\Sigma^*$ is a regular language with $\nNFA(L)=n$ then
$\nDFA(\up L)\leq 2^{n-2}+1$.

\noindent
2.\  Furthermore, for any $n>1$ there exists a regular language $L_n$ with
$\nNFA(L_n)=\nDFA(L_n)=n$ and $\nDFA(\up L_n)=\nUFA(\up L_n)=2^{n-2}+1$.
\end{proposition}
\begin{proof}
1.\ Let $A=(\Sigma,Q,\delta,I,F)$ be an $n$-state NFA recognizing $L=L(A)$. We can assume
$I\cap F=\emptyset$ and $\size{I\cup F}\geq 2$ otherwise $L$ contains
$\emptyword$ or is empty, resulting in a trivial $\up L$ with $\nDFA(\up L)=1$.

Since $A^\up$ has loops on all its states and for any letter, applying the
powerset construction yields a DFA where $P\step{a}P'$ implies $P\subseteq
P'$, hence any state $P$ reachable from $I$ includes $I$. Furthermore, if
$P$ is accepting, that is, $P\cap F\neq\emptyset$, and $P\step{a}P'$, then
$P'$ is accepting too, hence all accepting states recognize exactly
$\Sigma^*$ and are equivalent. Then there can be at most
$2^{\size{Q\setminus (I\cup F)}}$ states in the powerset automaton that are
both reachable and not accepting. To this we add 1 for the accepting states
since they are all equivalent and will be merged in the minimal DFA.
Finally $\nDFA(\up L)\leq 2^{n-2}+1$ since $\size{I\cup F}$ is at least $2$. \\

\noindent
2.\  To show that $2^{n-2}+1$ states may be necessary, we first consider the
case where $n=2$: taking $L_2=\{a\}$ over a 1-letter alphabet witnesses
both $\nDFA(L_n)=n=2$ and $\nDFA(\up L_n)=2^{n-2}+1=2$. Further, $\nUFA(\up
L_2)=2$ since clearly $\nNFA(\up L_2)>1$.

In the general case where $n>2$ we define $L_n \eqdef E_{n-2}$ where
\begin{equation*}
%\label{eq-def-Ek}
        E_k \eqdef \{a \, a~|~a\in\Sigma_{k}\} = \{ a_1\,a_1, \ldots,
        a_k\,a_k\}
\:.
\end{equation*}
In other words, $L_n$ contains all words consisting of two identical
letters from $\Sigma=\Sigma_{n-2}$. The minimal DFA recognizing $L_n$ has $n$
states, see Figure~\ref{fig-A1}. Now $\up L_n =\bigcup_{a\in
  \Sigma}\Sigma^*\cdot a\cdot \Sigma^*\cdot a\cdot\Sigma^*$, that is, $\up
L_n$ contains all words in $\Sigma^*$ where \emph{some letter reappears}. Thus
$\up L_n$ is the complement of the language we called $V_{n-2}$ above.
\begin{figure}[htbp]
\centering
\scalebox{0.85}{
  \begin{tikzpicture}[->,>=stealth',shorten >=1pt,node distance=6em,thick,auto,bend angle=30]
\tikzstyle{every state}=[minimum size=3em]
\node [state,initial] (p0) {$q_0$};
\node [right=4em of p0,state] (pi) {$q_i$};
\node [right=4em of pi,state,accepting] (pf) {$q_{n-1}$};
\node [above=3em of pi,state] (p1) {$q_1$};
\node [below=3em of pi,state] (pk) {$q_{n-2}$};
\node [above=0.6em of pi] (dummy1) {$\vdots$};
\node [above=0.6em of pk] (dummy2) {$\vdots$};
\path (p0) edge [pos=0.75] node {$a_1$} (p1);
\path (p0) edge [pos=0.70] node {$a_i$} (pi);
\path (p0) edge [pos=0.75,swap] node {$a_{n-2}$} (pk);
\path (p1) edge [pos=0.25] node {$a_1$} (pf);
\path (pi) edge [pos=0.30] node {$a_i$} (pf);
\path (pk) edge [pos=0.25,swap] node {$a_{n-2}$} (pf);
  \end{tikzpicture}
}%scalebox
%\vspace*{-1em}
\caption{$n$-state DFA recognizing $L_n=E_{n-2}=\{a_1\,a_1, a_2\,a_2, \ldots, a_{n-2}\,a_{n-2}\}$.}
\label{fig-A1}
%\vspace*{-1em}
\end{figure}
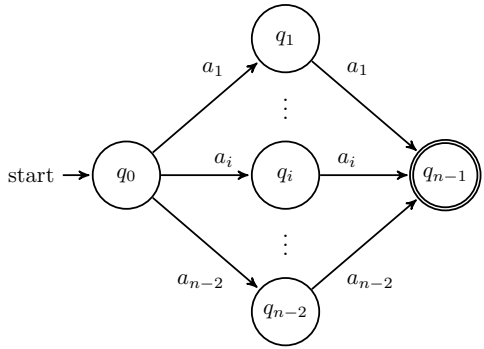
%%% Local Variables:
%%% fill-column: 999
%%% End:

The simplest way to recognize $\up L_n$ is via a DFA that records, in its
states, the set of letters previously read and accepts when one reappears.
This will use $2^{|\Sigma|}+1 = 2^{n-2}+1$ states, one for each
subset of previously read letters, to which one adds a single accepting
state. This DFA is minimal: given any two words $x$ and $y$ that reach
distinct states, one finds a $z$ such that $x\,z\in\up L_n$ and
$y\,z\not\in\up L_n$ or vice versa. We conclude that $\nDFA(\up
L_n)=2^{n-2}+1$ and deduce $\nNFA(L_n) = n$ (that is, we rule out $\nNFA(L_n) < n$)
from the first part of the lemma.

We refer to Proposition~\ref{prop-UFA-3} in Section~\ref{sec-ufa} for a proof that
recognizing $\up L_n$ requires $2^{n-2}+1$ states \emph{even for UFAs}.
\qed
\end{proof}

\begin{proposition}[State complexity of downward closure]
\label{prop-sc-downward-dfa}
1.\  If $L\subseteq\Sigma^*$ is recognized by an $n$-state NFA with a single
initial state then $\nDFA(\down L)\leq 2^{n-1}$.

\noindent
2.\  Furthermore, for any $n>1$ there exists a language $L'_n$ with
$\nNFA(L'_n)=\nDFA(L'_n)= n$ and $\nDFA(\down L'_n)=\nUFA(\down L'_n)=2^{n-1}$.
\end{proposition}
\begin{proof}
1.\  Assume that $L$ is recognized by $A=(\Sigma,Q,\delta,\{q_\init\},F)$,
an NFA where all states are reachable from $q_\init$, the single initial
state. From $A$ one derives an NFA $A^\down$ recognizing $\down L$ by adding
 $\emptyword$-transitions $q\step{\emptyword}q'$ for all pairs of states $q,q'$ such
that $q'$ is reachable from $q$. In particular, $A^\down$ contains
transitions $q_\init\step{\emptyword}q$ for all states $q\in Q$, and the
language accepted from $q$ is a subset of the language
accepted from $q_\init$. Hence, in the deterministic powerset automaton
obtained from $A^\down$, all states $P\subseteq Q$ that contain $q_\init$
are equivalent. This powerset automaton also
has up to $2^{n-1}-1$ nonempty states that do not contain $q_\init$. Thus
$1+2^{n-1}-1$ bounds the number of non-equivalent nonempty states in the
powerset automaton obtained from $A^\down$, showing $\nDFA(\down L)\leq 2^{n-1}$. \\

\noindent
2.\ To show that $2^{n-1}$ states are sometimes necessary, we assume $n>1$
and
let
 $L'_n \eqdef D_{n-1}$ where
\begin{equation*}
%\label{eq-def-Dk}
D_k \eqdef
\{x\in\Sigma_k^+~|~\forall i>1: x[i]\neq x[1]\} =
\bigcup_{a\in\Sigma_k}a\cdot\bigl(\Sigma_k\setminus a\bigr)^*
\:.
\end{equation*}
Thus $L'_n$ contains all words in $\Sigma_{n-1}^+$ where
\emph{the first letter does not reappear}.
The minimal DFA recognizing $L'_n$ has $n$ states, see Figure~\ref{fig-A2a}.
Every NFA for $L'_n$ has at least $n$ states, as shown by considering the following fooling
set:
\[
S = \left\{\begin{array}{l}
        (\emptyword, a_1 a_2 a_3 \cdots a_{n-1}),\:
        (a_1 , a_2 a_3 a_4 \cdots a_{n-1}),\:
        (a_2 , a_1 a_3 a_4 \cdots a_{n-1}),\:
        \\
        (a_3 , a_1 a_2 a_4 \cdots a_{n-1}),\:
        \ldots,\:
        (a_{n-1} , a_1 a_2 a_3 \cdots a_{n-2})
    \end{array}\right\}
\:.
\]

We now turn to
$\down L'_n = \{x~|~\exists a\in\Sigma_{n-1}:\forall i > 1: x[i]\neq a\}$.
That is, $\down L'_n$ contains all words $x$ such that the first suffix
$x[2..]$ does not use all letters. Equivalently $x\in\down L'_n$ if, and only if $x\in
L'_n$ or $x$ does not use all letters, that is, $\down L'_n$ is the union of
$L'_n$ and the complement of the language we called $U_{n-1}$ above.
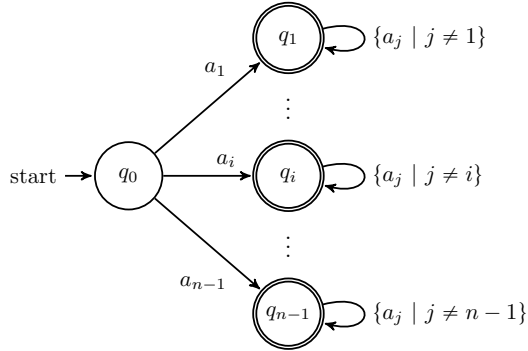
\begin{figure}[htbp]
\centering
\scalebox{0.85}{
  \begin{tikzpicture}[->,>=stealth',shorten >=1pt,node distance=6em,thick,auto,bend angle=30]
\tikzstyle{every state}=[minimum size=3em]
\node [state,initial] (p0) {$q_0$};
\node [right=4em of p0,state,accepting] (pi) {$q_i$};
\node [above=3em of pi,state,accepting] (p1) {$q_1$};
\node [below=3em of pi,state,accepting] (pk) {$q_{n-1}$};
\node [above=0.6em of pi] (dum1) {$\vdots$};
\node [above=0.6em of pk] (dum2) {$\vdots$};
\path (p0) edge [pos=0.75] node {$a_1$} (p1);
\path (p0) edge [pos=0.70] node {$a_i$} (pi);
\path (p0) edge [pos=0.75,swap] node {$a_{n-1}$} (pk);

\path (p1) edge[loop right] node {$\{a_j~|~j\neq 1\}$} (p1);
\path (pi) edge[loop right] node {$\{a_j~|~j\neq i\}$} (pi);
\path (pk) edge[loop right] node {$\{a_j~|~j\neq n-1\}$} (pk);

  \end{tikzpicture}
}%scalebox
%\vspace*{-1em}
\caption{$n$-state DFA recognizing $L'_n=D_{n-1}=\bigcup_{a\in\Sigma}a\cdot(\Sigma-\{a\})^*$ with $\size{\Sigma}=n-1$.}
\label{fig-A2a}
%\vspace*{-1em}
\end{figure}
%%% Local Variables:
%%% fill-column: 999
%%% End:

To show that $\nDFA(L'_n) = 2^{n-1}$, we start with a DFA $A$ that reads a
first letter and then starts recording which letters have been encountered
after the first one, in a manner similar to the construction of a DFA for
$U_\Sigma'$ in the proof of Lemma~\ref{lem-U-n-V}. All the states of $A$
are accepting but, in states of the form $\Sigma\setminus a$, the DFA
has no $a$-labelled transitions, hence $\Sigma$ is not a reachable state.
Finally $A$ has $1+2^{\size{\Sigma}}-1=2^{n-1}$ states: the initial state
reading the first letter, and one state for each strict subset of $\Sigma$.
This DFA is minimal: as in the previous proof, one checks that no two
states in $A$ are equivalent. Alternatively, one can refer to
Section~\ref{sec-ufa} where we prove -- see Proposition~\ref{prop-UFA-2} --
that recognizing $\down L'_n$ requires $2^{n-1}$ states \emph{even for
  UFAs}.
\qed
\end{proof}

\begin{remark}
\label{rem-down-NFA}
The condition of a single initial state in Proposition~\ref{prop-sc-downward-dfa}
cannot be lifted. It is possible to have $\nDFA(\down L)=\nUFA(\down
L)=2^n-1$ when $\nNFA(L)=n$. For example, the downward-closed language
$L=\Sigma_n^*\setminus U_n$ of all words that do not use all letters is
recognized by an $n$-state NFA (see Figure~\ref{fig-A2b}) but its minimal DFA
has $2^n-1$ states. In fact, any UFA recognizing $L$ has at least $2^n-1$
states (see Proposition~\ref{prop-UFA-1} in Section~\ref{sec-ufa}).
\end{remark}
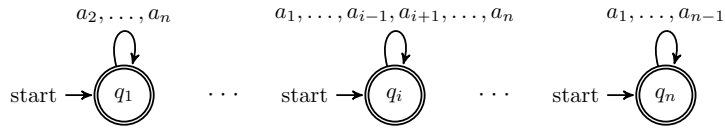
\begin{figure}[htbp]
\centering
\scalebox{0.85}{
  \begin{tikzpicture}[->,>=stealth',shorten >=1pt,node distance=6em,thick,auto,bend angle=30]
\tikzstyle{every state}=[minimum size=2.5em]

\node [state,initial,accepting] (q1) {$q_1$};
\node [right=2em of q1] (dum1) {\large $\bm{\cdots}$};
\node [right=5em of dum1,state,initial,accepting] (qi) {$q_i$};
\node [right=2em of qi] (dum2) {\large $\bm{\cdots}$};
\node [right=5em of dum2,state,initial,accepting] (qk) {$q_{n}$};

\path (q1) edge[loop above] node {$a_2,\ldots,a_n$} (q1);
\path (qi) edge[loop above] node {$a_1,\ldots,a_{i-1},a_{i+1},\ldots,a_n$} (qi);
\path (qk) edge[loop above] node {$a_1,\ldots,a_{n-1}$} (qk);

  \end{tikzpicture}
}%scalebox
%\vspace*{-1em}
\caption{$n$-state NFA recognizing $\Sigma_n^*\setminus U_n$.}
\label{fig-A2b}
%\vspace*{-1em}
\end{figure}
%%% Local Variables:
%%% fill-column: 999
%%% End:

The language families $(L_n)_{n\in\Nat}$ and $(L'_n)_{n\in\Nat}$ used to prove that the upper bounds given in
Propositions~\ref{prop-sc-upward-dfa} and~\ref{prop-sc-downward-dfa} are
tight use alphabets with a size linear in $n$.

It is known that the size of the alphabets matter for the state
complexity of closure operations. The automata witnessing tightness
in Figures~\ref{fig-A1} and~\ref{fig-A2a} use the smallest possible
alphabets. Okhotin showed that the $2^{n-2}+1$ state
complexity for $\up L$ cannot be achieved with an alphabet of a size smaller
than $n-2$, see~\cite[Lemma~4.4]{okhotin2010}.
We now prove a similar result for downward closures:
\begin{lemma}
  \label{lem-dc-alphabet-size}
For $n>2$, let $L \subseteq \Sigma^*$ be a regular language accepted by an $n$-state NFA with
a single initial state. If $\size{\Sigma}<n-1$ then $\nDFA(\down
L)<2^{n-1}$.
\end{lemma}
\begin{proof}
  We assume that $L$ is accepted by
 $A=(\Sigma,Q,\delta,\{q_\init\},F)$ with $\size{Q}=n$, that
  $\nDFA(\down L)= 2^{n-1}$ and deduce that
$\size{\Sigma}\geq n-1$.

We write $Q=\{q_\init,q_1,\ldots,q_{n-1}\}$ to denote the states of $A$. As we saw
in the proof of the first part of Proposition~\ref{prop-sc-downward-dfa},
the powerset automaton built from $A^\down$ can only have $2^{n-1}$ non-equivalent
reachable states if all non-empty subsets of $Q\setminus q_\init$ are reachable. Since $A^\down$ has
$\emptyword$-transitions doubling all transitions from $A$, it is possible to
construct the powerset automaton  with $Q$ as its initial
state. Then all edges $P\step{a}P'$ in the powerset automaton satisfy
$P\supseteq P'$. As a consequence, if $P\step{x}P'$ for some $x\in\Sigma^*$
then in particular one can pick $x$ with $\size{x}\leq \size{P\setminus P'}$.

Since every non-empty subset of $Q\setminus q_\init$ is reachable from $Q$ there is,
for every $i=1,\ldots,n-1$, some $x_i$ of length $1$ or $2$
such that $Q\step{x_i}Q\setminus q_\init,q_i$ (here $Q\setminus q,q'$ is shorthand for $Q\setminus \{q,q'\}$).
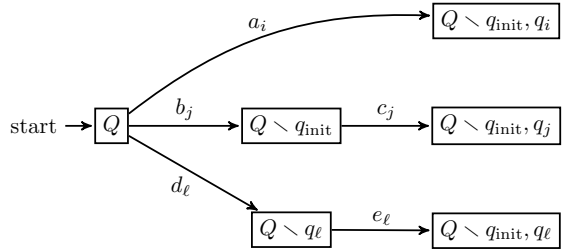
\begin{figure}[ht]
\centering
\scalebox{0.85}{
  \begin{tikzpicture}[->,>=stealth',shorten >=1pt,node distance=6em,thick,auto,bend angle=30]
\node [initial,rectangle,draw] (Q) {$Q$};
\node [right=5em of Q,rectangle,draw] (Qb) {$Q\setminus q_\init$};
\node [right=4em of Qb,rectangle,draw] (Qbc) {$Q\setminus q_\init, q_j$};
\node [above=3em of Qbc,rectangle,draw] (Qa) {$Q\setminus q_\init, q_i$};
\node [below=3em of Qb,rectangle,draw] (Qd)  {$Q\setminus q_\ell$};
\node [below=3em of Qbc,rectangle,draw] (Qde) {$Q\setminus q_\init, q_\ell$};
\path (Q) edge [bend left=20,pos=0.5] node [inner sep=1pt] {$a_i$} (Qa);
\path (Q) edge [pos=0.5] node [inner sep=2pt] {$b_j$} (Qb);
\path (Q) edge [pos=0.5,swap] node [inner sep=1pt] {$d_\ell$} (Qd);
\path (Qb) edge [pos=0.5] node [inner sep=2pt] {$c_j$} (Qbc);
\path (Qd) edge [pos=0.5] node [inner sep=2pt] {$e_\ell$} (Qde);
  \end{tikzpicture}
}%scalebox
\caption{A part of the powerset automaton of $A^\down$}
\label{fig-proof-powerset}
\end{figure}
%%% Local Variables:
%%% fill-column: 999
%%% End:
For a given $i$, there are three possible cases (see Figure~\ref{fig-proof-powerset}): $x_i=a_i$ is a
single letter (type 1), or $x_i$ is some $b_i\,c_i$ with
$Q\step{b_i}Q\setminus q_\init\step{c_i}Q\setminus q_\init, q_i$ (type 2), or $x_i$ is some
$d_i\,e_i$ with $Q\step{d_i}Q\setminus q_i\step{e_i}Q\setminus
q_\init, q_i$.

We now claim that the $a_i$'s for type-1 states, the $c_i$'s for type-2
states and the $d_i$'s for type-3 states are all distinct, hence
$\size{\Sigma}\geq n-1$.

Clearly the $a_i$'s and the $d_i$'s are pairwise distinct since they take $Q$ to
different states in the deterministic powerset automaton.
Similarly, the $c_i$'s are pairwise distinct, taking $Q\setminus q_\init$
to different states.

Assume now that $a_i=c_j$ for a type-1 $q_i$ and a type-2 $q_j$. Then
$Q\setminus q_\init\step{c_j}Q\setminus q_\init, q_j$ and
$Q\step{a_i(=c_j)}Q\setminus q_\init, q_i$, implying $q_i=q_j$ by
monotonicity of $\delta$ (the fact that $P_1\subseteq P_2$ implies
$\delta(P_1,a)\subseteq\delta(P_2,a)$ for any $a\in\Sigma$).

Similarly, assuming $d_\ell=c_j$ leads to $Q\setminus
q_\init\step{c_j}Q\setminus q_\init, q_j$ and
$Q\step{c_j(=d_\ell)}Q\setminus q_\ell$, implying $q_\ell=q_j$ by
monotonicity of $\delta$. Thus we can associate a distinct letter with each
state $q_1,\ldots,q_{n-1}$, which concludes the proof.
\qed
\end{proof}

In view of the above results, the main question is whether, \emph{in the case of a fixed alphabet}, exponential lower bounds
still apply for the (deterministic) state complexity of upward and downward closures. The 1-letter case is
degenerate since, when $\size{\Sigma}=1$, both $\nDFA(\up L)$ and $\nDFA(\down L)$ are  at most
$\nDFA(L)$. In the 3-letter case, exponential lower bounds for upward and
downward closures were shown by Okhotin~\cite{okhotin2010}.

In the critical 2-letter case, say $\Sigma=\{a,b\}$, an exponential lower bound for upward closure was
shown by H\'eam with the following witness: For $n>0$, let $L''_n = \{ a^i b \, a^{2j}
b \, a^i ~|~ i+j+1=n \}$. Then $\nDFA(L''_n) = (n+1)^2$, while $\nDFA(\up L''_n)
\geq \frac{1}{7}(\frac{1 + \sqrt{5}}{2})^n$ when $n\geq
4$~\cite[Proposition~5.11]{heam2002}.
Regarding downward closures for languages over a 2-letter alphabet, the
question was left open and we answer it in the next section.
%% REMOVED BECAUSE IT SEEMS IRRELEVANT -- phs
% However, the downward closure of these
% languages does not demonstrate a state blowup, in fact $\nDFA(\down L''_n) =
% n^2 + 3n - 1$ for $n \geq 2$.

\section{Exponential state complexity of closures in the 2-letter case}
%======================================================================
\label{sec-2-letter}

In this section we show an exponential lower bound for the state complexity
of downward closure in the case of
a two-letter alphabet. Interestingly,
the same lower bound for upward closure can be
proved using the same witnesses, but H\'eam already gave a stronger lower bound for upward closure~\cite{heam2002}.
\begin{theorem}[State complexity of closures with $\size{\Sigma}=2$]
\label{thm-ndfa-k=2}
The deterministic state complexity of downward closure for languages over the binary
alphabet $\Sigma=\{a,b\}$ is in
$2^{\Omega(n^{1/3})}$. The same result holds for upward closure.
\end{theorem}

We now prove the theorem. Fix $\Sigma=\{a,b\}$ and $n\in\Nat$.
Let
\[
H = \{n, n+1, \ldots, 2n\}
\:,
\]
and define morphisms $c, d : H^* \to \Sigma^*$ by
\begin{xalignat}{2}
\label{eq-def-c&d}
c(i) &\eqdef a^i \, b^{3n-i}
\:,
&
 d(i) &\eqdef c(i)\, c(i)
\:,
\end{xalignat}
for $i \in H$.
Note that $c(i)$ always has length $3n$, begins with at least $n$ $a$'s,
and ends with at least $n$ $b$'s. Let
\[
                      L_n \eqdef \{c(i)^n ~|~ i \in H\}
\:.
\]
The language $L_n$ is finite and contains $n+1$ words, each of length $3n^2$ so that
$\nDFA(L_n)$ is in $O(n^3)$. In fact, $\nDFA(L_n)=3n^3+1$.

In the rest of this section we show that, for $n$ even and strictly positive, both $\nDFA(\up
L_n)$ and $\nDFA(\down L_n)$ are greater than or equal to $\binom{n+1}{n/2}$. Since
$\binom{n+1}{n/2}\approx \frac{2^{n+3/2}}{\sqrt{\pi n}}$ and $\nDFA(L_n)=
3n^3+1$, the languages $(L_n)_{n=2,4,6,\ldots}$ witness the lower bound
claimed in Theorem~\ref{thm-ndfa-k=2}.

For each $i \in H$, let the morphisms $\eta_i, \theta_i : H^* \to (\Nat,
+)$ be defined by
\begin{xalignat*}{2}
 \eta_i(j) &\eqdef   \begin{cases}
                          1 & \text{ if } i \neq j\:, \\
                          2 & \text{ if } i = j\:,
                     \end{cases}
&
 \theta_i(j) &\eqdef \begin{cases}
                          2 & \text{ if } i \neq j\:, \\
                          1 & \text{ if } i = j\:.
                     \end{cases}
\end{xalignat*}
Thus for $\sigma=p_1\,p_2 \cdots p_s\in H^*$, $\eta_i(\sigma)$ is $s$ plus
the number of occurrences of $i$ in $\sigma$, while $\theta_i(\sigma)$ is
$2s$ minus the number of these occurrences of $i$.

\begin{lemma}\label{manydown}
Let $\sigma \in H^*$. The smallest $\ell$ such that $c(\sigma)$ is a
subword of
$c(i)^\ell$ is $\theta_i(\sigma)$.
\end{lemma}
\begin{proof}
We write $\sigma=p_1\,p_2\cdots p_s$ and prove the result by induction on
$s$. The case of $s=0$ is trivial. For the case of $s=1$, note that for
any $p_1$ and $i$, $c(p_1) \embeds d(i)=c(i)^2$ and that $c(p_1) \embeds
c(i)$ if and only if $p_1=i$.

Assume now that $s>1$, write $\sigma = \sigma' p_s$ and let
$\ell'=\theta_i(\sigma')$. By the induction hypothesis, $c(\sigma') \not
\embeds c(i)^{\ell'-1}$ and $c(\sigma') \embeds
c(i)^{\ell'}=c(i)^{\ell'-1}a^i b^{3n-i}$. Write now $c(i)^{\ell'}=w\,v$
where $w$ is the shortest prefix of $c(i)^{\ell'}$ with $c(\sigma') \embeds
w$. Since $c(\sigma')$ ends with some $b$ that only embeds in the suffix
$a^i b^{3n-i}$ of $c(i)^{\ell'}$, $v$ is necessarily $b^r$ for some
$r$. So, for all $z \in \Sigma^*$, $c(p_s) \embeds z$ if and only if
$c(p_s) \embeds v\,z$. We have $c(p_s) \embeds c(i)^{\theta_i(p_s)}$ and
$c(p_s) \not \embeds v\, c(i)^{\theta_i(p_s)-1}$. Noting that $\sigma =
\sigma' p_s$, we get $c(\sigma) \embeds c(i)^{\theta_i(\sigma)}$ and
$c(\sigma) \not \embeds c(i)^{\theta_i(\sigma)-1}$.
\qed
\end{proof}

We now derive the announced lower bound on $\nDFA(\down L_n)$. Recall that
$n$ is even and strictly positive. For every subset $X$ of $H$ of size $n/2$, let
$w_X \in \Sigma^*$ be defined as follows: let the elements of $X$ be $p_1 <
p_2 < \cdots < p_{n/2}$ and let
\[
w_X \eqdef c(p_1 p_2 \cdots p_{n/2} )
\:.
\]
Note that $\theta_i(p_1 p_2 \cdots p_{n/2} ) = n$ if $i \notin X$ and
$\theta_i(p_1 p_2 \cdots p_{n/2} ) = n-1$ if $i \in X$.

\begin{lemma}
\label{lem-XY-down}
Let $X$ and $Y$ be subsets of $H$ of size $n/2$ with $X \neq Y$. There
exists a word $v \in \Sigma^*$ such that $w_X v \in \down L_n$ and $w_Y v
\notin \down L_n$.
\end{lemma}
\begin{proof}
Let $i \in X \setminus Y$. Let $v = c(i)$.
By Lemma~\ref{manydown}, $w_X \embeds c(i)^{n-1}$, and so $w_Xv \embeds
c(i)^n$, hence $w_Xv \in \down L_n$.

By Lemma~\ref{manydown}, the smallest $\ell$ such that $w_Y v \embeds
c(i)^\ell$ is $n+1$. Similarly, for $j \neq i$, the smallest $\ell$ such
that $w_Y v \embeds c(j)^\ell$ is at least $n-1+2 = n+1$ (at least $n-1$
for the $w_Y$
factor and $2$ for the $v$ factor). So $w_Y v \notin
\down L_n$.
\qed
\end{proof}

This shows that for any DFA $A=(\Sigma,Q,\delta,q_1,F)$ recognizing $\down
L_n$, the states of the form $\delta(q_1,w_X)$ for a subset $X \subseteq H$
with $\size{X} = n/2$ are all distinct.
Thus $A$ has at least $\binom{n+1}{n/2}$
states as claimed.
\\

For $\nDFA(\up L_n)$, the reasoning is similar:

\begin{lemma}\label{upmany} %\label{lem-single}
\begin{enumerate}
 \item For $i, j \in H$, the longest prefix of $c(i)^\omega$ that is a subword of $d(j)
= c(j)\,c(j)$ is $c(i)$ if $i \neq j$ and $c(i)\,c(i)$ if $i=j$.
\item Let $\sigma \in H^*$. For all $i \in H$, the longest prefix of
$c(i)^\omega$ that is a subword of $d(\sigma)$ is $c(i)^{\eta_i(\sigma)}$.
\end{enumerate}

\end{lemma}
\begin{proof}
\begin{enumerate}
\item The statement is trivial when $i=j$, so we now assume $i\neq j$.
  Equation~\eqref{eq-def-c&d} entails $c(i)\subword c(j)c(j)$ since $n\leq
  i,j\leq 2n$. It remains to show that no longer
  prefix of $c(i)^\omega$ embeds in $c(j)c(j)$, that is, that $c(i)a\not\subword c(j)c(j)$.
  But this is clear when one considers the leftmost embedding of $c(i)a$ in
  $c(j)c(j)$: this is illustrated by Figure~\ref{fig-lem-single} in the case
  of $i>j$, the case of $i<j$ being similar.

\item By induction on the length of $\sigma$, as above. \qed
\end{enumerate}
\end{proof}
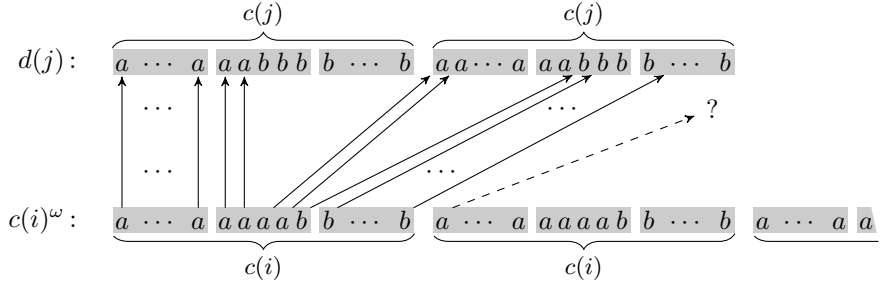
\begin{figure}[htbp]
\centering
\scalebox{1.0}{
\newcommand{\BO}[1]{\raisebox{0pt}[0.7em][0pt]{#1}}
\begin{tikzpicture}[->,>=stealth',shorten >=1pt,node distance=6em,thin,auto,bend angle=30]

\clip (0,1.22) -- (11.7,1.22) -- (10.9,5) -- (0,5) -- (0,1.22) -- cycle;

{\tikzstyle{every node}=[fill=black!20,text width=2.9em,text height=0.3em,align=center]
\node (p1) at (2.03,4.18) {};
\node [right=.3em of p1] (p2) {};
\node [right=.3em of p2] (p3) {};
\node [right=.7em of p3] (p4) {};
\node [right=.3em of p4] (p5) {};
\node [right=.3em of p5] (p6) {};
\node [below=5em of p1] (q1) {};
\node [right=.3em of q1] (q2) {};
\node [right=.3em of q2] (q3) {};
\node [right=.7em of q3] (q4) {};
\node [right=.3em of q4] (q5) {};
\node [right=.3em of q5] (q6) {};
\node [right=.7em of q6] (q7) {};
\node [right=.3em of q7] (q8) {};%needed for brace positioning but clipped out
\node [right=.3em of q8] (q9) {};%same
}
{\tikzstyle{every node}=[minimum size=0.01em,inner sep=0.1em]
\node at ($(p1.west)!0.1!(p1.east)$) (p1a) {\BO{$a$}};
\node at ($(p1.west)!0.3!(p1.east)$) (p1b) {};
\node at ($(p1.west)!0.5!(p1.east)$) (p1c) {\BO{$\cdots$}};
\node at ($(p1.west)!0.7!(p1.east)$) (p1d) {};
\node at ($(p1.west)!0.9!(p1.east)$) (p1e) {\BO{$a$}};
\node at ($(p2.west)!0.1!(p2.east)$) (p2a) {\BO{$a$}};
\node at ($(p2.west)!0.3!(p2.east)$) (p2b) {\BO{$a$}};
\node at ($(p2.west)!0.5!(p2.east)$) (p2c) {\BO{$b$}};
\node at ($(p2.west)!0.7!(p2.east)$) (p2d) {\BO{$b$}};
\node at ($(p2.west)!0.9!(p2.east)$) (p2e) {\BO{$b$}};
\node at ($(p3.west)!0.1!(p3.east)$) (p3a) {\BO{$b$}};
\node at ($(p3.west)!0.3!(p3.east)$) (p3b) {};
\node at ($(p3.west)!0.5!(p3.east)$) (p3c) {\BO{$\cdots$}};
\node at ($(p3.west)!0.7!(p3.east)$) (p3d) {};
\node at ($(p3.west)!0.9!(p3.east)$) (p3e) {\BO{$b$}};
\node at ($(p4.west)!0.1!(p4.east)$) (p4a) {\BO{$a$}};
\node at ($(p4.west)!0.3!(p4.east)$) (p4b) {\BO{$a$}};
\node at ($(p4.west)!0.6!(p4.east)$) (p4c) {\BO{$\cdots$}};
\node at ($(p4.west)!0.7!(p4.east)$) (p4d) {};
\node at ($(p4.west)!0.9!(p4.east)$) (p4e) {\BO{$a$}};
\node at ($(p5.west)!0.1!(p5.east)$) (p5a) {\BO{$a$}};
\node at ($(p5.west)!0.3!(p5.east)$) (p5b) {\BO{$a$}};
\node at ($(p5.west)!0.5!(p5.east)$) (p5c) {\BO{$b$}};
\node at ($(p5.west)!0.7!(p5.east)$) (p5d) {\BO{$b$}};
\node at ($(p5.west)!0.9!(p5.east)$) (p5e) {\BO{$b$}};
\node at ($(p6.west)!0.1!(p6.east)$) (p6a) {\BO{$b$}};
\node at ($(p6.west)!0.3!(p6.east)$) (p6b) {};
\node at ($(p6.west)!0.5!(p6.east)$) (p6c) {\BO{$\cdots$}};
\node at ($(p6.west)!0.7!(p6.east)$) (p6d) {};
\node at ($(p6.west)!0.9!(p6.east)$) (p6e) {\BO{$b$}};
\node at ($(q1.west)!0.1!(q1.east)$) (q1a) {\BO{$a$}};
\node at ($(q1.west)!0.3!(q1.east)$) (q1b) {};
\node at ($(q1.west)!0.5!(q1.east)$) (q1c) {\BO{$\cdots$}};
\node at ($(q1.west)!0.7!(q1.east)$) (q1d) {};
\node at ($(q1.west)!0.9!(q1.east)$) (q1e) {\BO{$a$}};
\node at ($(q2.west)!0.1!(q2.east)$) (q2a) {\BO{$a$}};
\node at ($(q2.west)!0.3!(q2.east)$) (q2b) {\BO{$a$}};
\node at ($(q2.west)!0.5!(q2.east)$) (q2c) {\BO{$a$}};
\node at ($(q2.west)!0.7!(q2.east)$) (q2d) {\BO{$a$}};
\node at ($(q2.west)!0.9!(q2.east)$) (q2e) {\BO{$b$}};
\node at ($(q3.west)!0.1!(q3.east)$) (q3a) {\BO{$b$}};
\node at ($(q3.west)!0.3!(q3.east)$) (q3b) {};
\node at ($(q3.west)!0.5!(q3.east)$) (q3c) {\BO{$\cdots$}};
\node at ($(q3.west)!0.7!(q3.east)$) (q3d) {};
\node at ($(q3.west)!0.9!(q3.east)$) (q3e) {\BO{$b$}};
\node at ($(q4.west)!0.1!(q4.east)$) (q4a) {\BO{$a$}};
\node at ($(q4.west)!0.3!(q4.east)$) (q4b) {};
\node at ($(q4.west)!0.5!(q4.east)$) (q4c) {\BO{$\cdots$}};
\node at ($(q4.west)!0.7!(q4.east)$) (q4d) {};
\node at ($(q4.west)!0.9!(q4.east)$) (q4e) {\BO{$a$}};
\node at ($(q5.west)!0.1!(q5.east)$) (q5a) {\BO{$a$}};
\node at ($(q5.west)!0.3!(q5.east)$) (q5b) {\BO{$a$}};
\node at ($(q5.west)!0.5!(q5.east)$) (q5c) {\BO{$a$}};
\node at ($(q5.west)!0.7!(q5.east)$) (q5d) {\BO{$a$}};
\node at ($(q5.west)!0.9!(q5.east)$) (q5e) {\BO{$b$}};
\node at ($(q6.west)!0.1!(q6.east)$) (q6a) {\BO{$b$}};
\node at ($(q6.west)!0.3!(q6.east)$) (q6b) {};
\node at ($(q6.west)!0.5!(q6.east)$) (q6c) {\BO{$\cdots$}};
\node at ($(q6.west)!0.7!(q6.east)$) (q6d) {};
\node at ($(q6.west)!0.9!(q6.east)$) (q6e) {\BO{$b$}};
\node at ($(q7.west)!0.1!(q7.east)$) (q7a) {\BO{$a$}};
\node at ($(q7.west)!0.3!(q7.east)$) (q7b) {};
\node at ($(q7.west)!0.5!(q7.east)$) (q7c) {\BO{$\cdots$}};
\node at ($(q7.west)!0.7!(q7.east)$) (q7d) {};
\node at ($(q7.west)!0.9!(q7.east)$) (q7e) {\BO{$a$}};
\node at ($(q8.west)!0.1!(q8.east)$) (q8a) {\BO{$a$}};
\node at ($(q8.west)!0.3!(q8.east)$) (q8b) {\BO{$a$}};
\node at ($(q8.west)!0.5!(q8.east)$) (q8c) {\BO{$a$}};
\node at ($(q8.west)!0.7!(q8.east)$) (q8d) {\BO{$a$}};
\node at ($(q8.west)!0.9!(q8.east)$) (q8e) {\BO{$b$}};
\node at ($(q9.west)!0.1!(q9.east)$) (q9a) {\BO{$b$}};
\node at ($(q9.west)!0.3!(q9.east)$) (q9b) {};
\node at ($(q9.west)!0.5!(q9.east)$) (q9c) {\BO{$\cdots$}};
\node at ($(q9.west)!0.7!(q9.east)$) (q9d) {};
\node at ($(q9.west)!0.9!(q9.east)$) (q9e) {\BO{$b$}};
}

\draw[-,decorate,decoration={brace,amplitude=6pt}] (p1.north west) -- (p3.north east) node [midway,yshift=.5em] {\small $c(j)$};
\draw[-,decorate,decoration={brace,amplitude=6pt}] (p4.north west) -- (p6.north east) node [midway,yshift=.5em] {\small $c(j)$};
\draw[-,decorate,decoration={brace,amplitude=6pt}] (q3.south east) -- (q1.south west) node [midway,yshift=-.5em] {\small $c(i)$};
\draw[-,decorate,decoration={brace,amplitude=6pt}] (q6.south east) -- (q4.south west) node [midway,yshift=-.5em] {\small $c(i)$};
\draw[-,decorate,decoration={brace,amplitude=6pt}] (q9.south east) -- (q7.south west) node [midway,yshift=-.5em] {\small $c(i)$};

\node at ($(q1c)!0.3!(p1c)$) (dots1) {$\cdots$};
\node at ($(q1c)!0.7!(p1c)$) (dots2) {$\cdots$};
\node at ($(q3c)!0.3!(p5e)$) (dots3) {$\cdots$};
\node at ($(q3c)!0.7!(p6a)$) (dots4) {$\cdots$};
\node[right=4em of dots4] (qm) {?};

\path
(q1a.north) edge (p1a.south)
(q1e.north) edge (p1e.south)
(q2a.north) edge (p2a.south)
(q2b.north) edge (p2b.south)
(q2c.north east) edge (p4a.south west)
(q2d.north east) edge (p4b.south west)
(q2e.north east) edge (p5c.south west)
(q3a.north east) edge (p5d.south west)
(q3e.north east) edge (p6c.south west)
(q4a.north east) edge [dashed] (qm)
;

\node [left=1em of p1,align=flush right] (lp) {$d(j)\,$:};
\node [left=1em of q1,align=flush right] (lq) {$c(i)^\omega\,$:};

\end{tikzpicture}
}%scalebox
%\vspace*{-1em}
\caption{Case ``$i>j$'' in Lemma~\ref{upmany} (here with $n=5$, $i=n+4$ and $j=n+2$).}
\label{fig-lem-single}
%\vspace*{-1em}
\end{figure}
%%% Local Variables:
%%% fill-column: 999
%%% End:

Recall that $n$ is even and strictly positive. For every subset $X$ of $H$ of size $n/2$, let $w'_X
\in \Sigma^*$ be defined as follows: let the elements of $X$ be $p_1 < p_2
< \cdots < p_{n/2}$ and let
\[
w'_X \eqdef d(p_1 p_2 \cdots p_{n/2} ) = c(p_1 p_1 p_2 p_2 \cdots p_{n/2} p_{n/2})
\:.
\]

\begin{lemma}
\label{lem-XY-up}
Let $X$ and $Y$ be subsets of $H$ of size $n/2$ with $X \neq Y$. There
exists a word $v \in \Sigma^*$ such that $w'_X v \in \up L_n$ and $w'_Y v
\notin \up L_n$.
\end{lemma}
\begin{proof}
Let $i \in X \setminus Y$.  Let $v = c(i)^{n - (n/2 +1)} = c(i)^{n/2-1}$.
By Lemma~\ref{upmany}, $c(i)^{n/2+1} \embeds w'_X$, thus $c(i)^n \embeds
w'_X v$, hence $w'_X v \in \up L_n$.

We now show that $w'_Y v \notin \up L_n$. By Lemma~\ref{upmany}, the longest prefix of $c(i)^n$ that embeds in $w'_Y
v$ is a prefix of $c(i)^\ell$ where $\ell = n/2 + n/2 - 1 = n-1$. Thus
$w'_Y v \notin \up c(i)^n$. For $j \neq i$, we show $c(j)^n \not \subword
w'_Y v$ by contradiction. Suppose $c(j)^n \subword w'_Y v$. The longest
prefix of $c(j)^n$ that is a subword of $w'_Y$ is a prefix of
$c(j)^{n/2+1}$. Thus $c(j)^{n/2 - 1} \subword v$. But $c(j)^{n/2-1}$ and
$v$ are different words of the same length, so this is not possible. Thus
$c(j)^n \not \subword w'_Y v$. Finally $w'_Y v \notin \up L_n$.
\qed
\end{proof}
With Lemma~\ref{lem-XY-up} we reason exactly as we did for
$\nDFA(\down L_n)$ after Lemma~\ref{lem-XY-down} and conclude
that $\nDFA(\up L_n)$ is at least $\binom{n+1}{n/2}$.

%% Local Variables:
%% ispell-check-comments: nil
%% ispell-local-dictionary: "english"
%% fill-column: 75
%% TeX-master: "main"
%% End:

% LocalWords:  DFAs eam UFAs

\section{State complexity of interiors}
%======================================
\label{sec-interior}

Recall Equation~\eqref{eq-interiors-are-duals} expressing interiors with
closures and complements. Since complementation of DFAs does not increase
the number of states, except perhaps adding a single state if we start with an incomplete DFA,
the state complexity of interiors, seen as DFA to DFA
operations, is essentially the same as the state complexity of closures modulo
swapping of up and down.
\\

The remaining question is the \emph{nondeterministic state complexity} of
interiors, now seen as NFA to NFA operations.
For this, Equation~\eqref{eq-interiors-are-duals} provides an obvious $2^{2^n}$
upper bound on the nondeterministic state complexity of both upward and downward
interiors, simply by combining the powerset construction for complementation
and the results of Section~\ref{sec-closure}.
Note that this procedure yields DFAs for the interiors while we are happy
to accept NFAs if this improves the state complexity.

In the rest of this section, we prove that the nondeterministic state
complexity of upward and downward interiors is in $2^{2^{\Theta(n)}}$.
Sections~\ref{ssec-psi-upper-bounds} and~\ref{ssec-lower-bounds-both-interiors}
establish the upper and lower bounds, respectively. Section~\ref{ssec-afas}
mentions the consequences on representations based on alternating automata.

\subsection{Upper bounds for interiors and the approximation problem}
%====================================================================
\label{ssec-psi-upper-bounds}

We first give an upper bound for the state complexity of interiors that
slightly improves on the obvious $2^{2^n}$ upper bound.
For this we adapt a technique from~\cite{conway71,lombardy2008} and rely on
the fact, already used in~\cite[Theorem~6.1]{chandra76}, that the state
complexity of a positive Boolean combination of left-quotients of some
regular language $L$ is at most $\psi(\nNFA(L))$.
\begin{proposition}
\label{prop-interiors-nsc}
Let $L\subseteq \Sigma^*$ be a regular language with $\nNFA(L)=n$.
Then $\nDFA(\upint L)<\psi(n)$ and
 $\nDFA(\downint L)<\psi(n)$.
\end{proposition}
\begin{proof}
We handle both interiors in a uniform way.

Let $K_0$ and $K_1,\ldots,K_p$ be arbitrary languages in $\Sigma^*$ (these need not be regular). With the $K_i$'s
we associate an alphabet $\Gamma=\{b_1,\ldots,b_p\}$ and a substitution
$\sigma$ given inductively by $\sigma(\emptyword) \eqdef K_0$ and
$\sigma(w \, b_i) \eqdef \sigma(w)\cdot K_i$.
With a language $L\subseteq \Sigma^*$, we associate the language
$W\subseteq\Gamma^*$ defined by
\begin{equation}
\label{eq-def-W}
              W\eqdef \{x\in\Gamma^*~|~\sigma(x)\subseteq L\}
\:.
\end{equation}
\begin{claim}
If $L$ is regular then $W$ is regular.
\end{claim}
To prove this first claim,
assume $A_1=(\Sigma,Q,\delta_1,I_1,F_1)$ is an $n$-state NFA recognizing $L$. Using
the powerset construction, one obtains a DFA
$A_2=(\Sigma,Q_2,\delta_2,i_2,F_2)$ recognizing $L$. We have as usual $Q_2 = 2^Q$,
with typical elements $S,S',\ldots$,
$\delta_2$ given by $\delta_2(S, a) = \bigcup_{q \in S} \delta_1(q,a)$,
$i_2 = I_1$, and $F_2 = \{S ~|~ S \cap F_1 \neq \emptyset\}$.

From $A_2$ we now derive a DFA $A_3 = (\Gamma, Q_3, \delta_3, i_3, F_3)$ given by
 $Q_3 = 2^{Q_2}$, with typical elements $U,U',\ldots$;
 $\delta_3(U, b_j) =  \{ \delta_2(S, z) ~|~ S \in U, z \in K_j \}$;
 $i_3 = \{\delta_2(i_2, z) ~|~ z \in K_0\}$;
and $F_3 = 2^{F_2} = \{U ~|~ U \subseteq F_2\}$.

The intention is that $A_3$ will recognize $W$, so let us check, using
induction on $w\in\Gamma^*$, that
$\delta_3(i_3, w) = \{\delta_2(i_2, z) ~|~ z \in \sigma(w) \}$:
For the base case, one has $\delta_3(i_3, \emptyword) =
i_3 = \{\delta_2(i_2, z) ~|~ z \in K_0\}$ by definition, and
$\sigma(\emptyword) = K_0$. For the inductive case, one has
\begin{align*}
    & \delta_3(i_3, w \, b_j) =  \delta_3(\delta_3(i_3, w), b_j) \\
=\: & \delta_3(\{\delta_2(i_2, z) ~|~ z \in \sigma(w) \}, b_j) & \text{(induction hypothesis)} \\
=\: & \bigl\{\delta_2(S, z') ~|~ S \in \{\delta_2(i_2, z) ~|~ z \in \sigma(w) \} , z' \in K_j \bigr\}\!\!\!\!\!\!\!\!\!\!\!\!\! & \text{(definition of $\delta_3$)} \\
=\: & \{ \delta_2(\delta_2(i_2, z), z') ~|~ z \in \sigma(w), z' \in \sigma(b_j) \} & \text{(rearrange, use $\sigma(b_j) = K_j$)} \\
=\: & \{ \delta_2(i_2, z'') ~|~ z'' \in \sigma(w \, b_j) \}\:. &
\end{align*}
Now, for all $w \in \Gamma^*$, one has
\begin{align*}
  w \in W &\iff \sigma(w) \subseteq L
  &\text{(definition of $W$)}
\\
  &\iff \forall z \in \sigma(w):\delta_2(i_2,z) \in F_2
  &\text{(since $A_2$ recognizes $L$)}
\\
 &\iff \{\delta_2(i_2, z) ~|~ z \in \sigma(w)\} \subseteq F_2
\\
  &\iff \delta_3(i_3, w) \in F_3 \:.
  &\text{(as just shown)}
\end{align*}
This proves that $A_3$ recognizes $W$. 
In particular, $W$ is
regular as claimed.

\begin{claim}
$\nDFA(W)<\psi(n)$.
\end{claim}
The DFA $A_3$ that recognizes $W$ has $\size{Q_3}=2^{2^n}$ states. We now
examine our construction more closely to detect equivalent states in $A_3$. Observe that
the powerset construction for $A_2$ in terms of $A_1$ is ``existential'',
that is, a state of $A_2$ is accepting if and only if at least one of its
constituent states from $A_1$ is accepting. In contrast, the powerset
construction for $A_3$ in terms of $A_2$ is ``universal'', that is, a state
of $A_3$ is accepting if and only if all of its constituent states from $A_2$
are accepting. Suppose $S, S' \in Q_2$ are two states of $A_2$ with $S
\subseteq S'$. Then if
some word is accepted by $A_2$ starting from $S$, it is also accepted
starting from $S'$. If a state of $A_3$ contains both $S$ and $S'$, then $S$
already imposes a stronger constraint than $S'$, and so $S'$ can be
eliminated. We make this precise below:

Define an equivalence relation $\equiv$ on $Q_3$ as follows:
\[
U \equiv U'
\:\equivdef\:
( \forall S \in U: \exists S' \in U': S' \subseteq S)
\land (\forall S' \in U': \exists S \in U: S \subseteq S')
\:.
\]
We now claim that, in $A_3$, $\equiv$-equivalent states accept the same
language. First $U\equiv V$ and $U \in F_3$ imply $V \in F_3$ since for
any $S' \in V$, there is $S \in U$ with $S\subseteq S'$, and since $S \in
F_2$, also $S' \in F_2$. Furthermore  $U\equiv V$ and $b_j\in\Gamma$ imply $\delta(U, b_j) \equiv \delta(V, b_j)$: each element of $\delta_3(U,b_j)$ is
some $\delta_2(S, z)$ with  $S \in U$ and $z \in K_j$. There exists $S' \in
V$ such that $S' \subseteq S$, and then $\delta_2(S', z)$ belongs to
$\delta_3(V, b_j)$ and is a subset of $\delta_2(S, z)$ because $\delta_2$
is monotone in its first argument. The reasoning in the reverse direction
is similar.

Thus we can quotient the DFA $A_3$ by $\equiv$ to get an equivalent DFA
recognizing $W$. Further, we can remove (the equivalence class of) the sink
state $\{\emptyset\}$, so that $\nDFA(W)<\size{Q_3/{\equiv}}$.

Let us now show that $\size{Q_3/{\equiv}}$ is exactly $\psi(\size{Q})$.
A state $U \in Q_3$ is called an \emph{antichain} if it does not contain
some $S,S'\in Q_2$ with $S \subsetneq S'$.  Every
 $U \in Q_3$ is $\equiv$-equivalent to the antichain $U_\text{min}$
obtained by retaining only the elements of $U$ that are  minimal by inclusion.
Further, two distinct antichains cannot be $\equiv$-equivalent.
Thus the number of equivalence classes in $Q_3/{\equiv}$
is exactly the number of
subsets of $2^Q$ which are antichains, and this is the Dedekind number
$\psi(n)$, see~\cite{kleitman69}.
This shows $\nDFA(W)<\psi(n)$, completing the proof of
our second claim.
\\

We may now instantiate the above construction for the upward and downward
interiors. Choose alphabets $\Sigma = \Gamma = \{b_1, \ldots, b_k\}$ and
let $K_0 = \Sigma^*$ and $K_i = \Sigma^*b_i\Sigma^*$. Then
Equation~\eqref{eq-def-W} yields $W=\upint(L)$ and we deduce
$\nDFA(\upint(L))<\psi(n)$. Letting now $K_0 = \{\emptyword\}$ and $K_i =
\{b_i, \emptyword\}$ yields $W=\downint(L)$ and again we deduce
$\nDFA(\downint(L))<\psi(n)$. This concludes the proof of
Proposition~\ref{prop-interiors-nsc}.
\qed
\end{proof}

\begin{remark}
In the usual setting -- see~\cite[Section~6]{lombardy2008} -- $W$ is
defined with $\sigma(\emptyword)=\{\emptyword\}$ and there is no need for
$K_0$. The idea is that $W$ is the best under-approximation of $L$ by sums
of products of $K_i$'s, and Conway showed that if $L$ is regular then $W$
is too~\cite{conway71}. We allowed $\sigma(\emptyword)=K_0$ to account
directly for upward interiors.
\end{remark}

\subsection{Lower bounds for interiors}
%==============================================
\label{ssec-lower-bounds-both-interiors}

\begin{proposition}[Downward interior]
\label{prop-nfa-downint}
There exists a family of languages $(L_n)_{n=3,4,\ldots}$ with $\nNFA(L_n)\leq n$ and
$\nNFA(\downint L_n) = 2^{2^{\left\lfloor \frac{n-3}{2}\right\rfloor}}$.
\end{proposition}
\begin{proof}
Fix $n\geq 3$ and let $\ell=\left\lfloor \frac{n-3}{2}\right\rfloor$. We
let $\Sigma = \{0,1,2,\ldots,2^\ell-1\}$, so that $\size{\Sigma}=2^\ell$.
Let
\[
L_n \eqdef \Sigma^* \setminus \{ a\,a ~|~ a \in \Sigma \}
= \{a\,b ~|~ a,b \in \Sigma, a \neq b\} \cup \{ w \in\Sigma^* ~|~  \size{w} \neq 2 \}
\:.
\]
That is, $L_n$ contains all words over $\Sigma$ consisting of two
different letters and all words whose length is not $2$.
\\

We first prove that $\nNFA(L_n)\leq 2\ell+3\leq n$: Two letters in
$\Sigma$, viewed as $\ell$-bit sequences, are distinct if and only if they
differ in at least one bit. Figure~\ref{fig-A3} displays an NFA recognizing $\{a\,b
\in\Sigma^2 ~|~ a \neq b\}$ with $2\ell+2$ states: the idea is that the NFA
reads $a$, guesses the position of a bit where $a$ and $b$ differ, records
the value of $a$'s corresponding bit and checks $b$'s bit at that position.
\begin{figure}[htb]
\centering
\scalebox{0.75}{
  \begin{tikzpicture}[->,>=stealth',shorten >=1pt,node distance=5em,thick,auto,bend angle=30]
\node [shape=rectangle] (dot+) {$\vdots$};
\node [right of=dot+,shape=rectangle] (dot-) {$\vdots$};
\node [above of=dot+,state] (q1+) {$2^+$};
\node [above of=dot-,state] (q1-) {$2^-$};
\node [above of=q1+,state] (q0+) {$1^+$};
\node [above of=q1-,state] (q0-) {$1^-$};
\node [below of=dot+,state] (ql+) {$\ell^+$};
\node [below of=dot-,state] (ql-) {$\ell^-$};
\node [left of=dot+,node distance=16em,state,initial] (qin) {in};
\node [right of=dot-,node distance=16em,state,accepting] (qf) {f\/i};

\path[every node/.style={auto=false,sloped,fill=black!10}]
 (qin) edge [bend left=20]  node {$1,3,5,\ldots$} (q0+)
 (qin) edge [bend left=10]  node {$0,2,4,\ldots$} (q0-)
 (q0+) edge [bend left=10]  node {$0,2,4,\ldots$} (qf)
 (q0-) edge [bend left=20]  node {$1,3,5,\ldots$} (qf)
 (qin) edge [bend left=0]   node {$2,3,6,7,\ldots$} (q1+)
 (qin) edge [bend right=15] node {$0,1,4,5,\ldots$} (q1-)
 (q1+) edge [bend right=15] node {$0,1,4,5,\ldots$} (qf)
 (q1-) edge [bend left=0]   node {$2,3,6,7,\ldots$} (qf)
 (qin) edge [bend left=3]   node {$0,1,\ldots,2^{\ell-1}-1$} (ql-)
 (qin) edge [bend right=15] node {$2^{\ell-1},2^{\ell-1}+1,\ldots,2^\ell-1$} (ql+)
 (ql-) edge [bend right=15] node {$2^{\ell-1},2^{\ell-1}+1,\ldots,2^\ell-1$} (qf)
 (ql+) edge [bend left=3]   node {$0,1,\ldots,2^{\ell-1}-1$} (qf)
;
  \end{tikzpicture}
}%scalebox
%\vspace*{-1em}
\caption{NFA recognizing $\{a\,b~|~a,b\in\Sigma_{2^\ell}, a\neq b\}$ with $2\ell+2$ states.}
\label{fig-A3}
%\vspace*{-1em}
\end{figure}
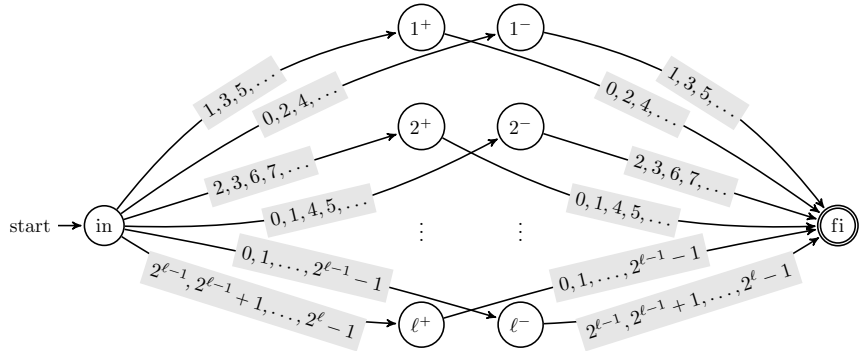
%%% Local Variables:
%%% fill-column: 999
%%% End:
We then modify this NFA so that it also accepts all words whose length
is not $2$. This can be done by adding a single new state and appropriate transitions, and making all
original states accepting. The resulting NFA has $2\ell + 3$ states.
\\

It remains to prove that $\nNFA(\downint L_n)= 2^{2^\ell}$, but
$\downint L_n$ consists of all words in $\Sigma^*$ of pairwise
distinct letters, the language called $V_\Sigma$ in
Lemma~\ref{lem-U-n-V} where we showed $\nNFA(V_\Sigma)= 2^{\size{\Sigma}}
= 2^{2^\ell}$.
\qed
\end{proof}

\begin{proposition}[Upward interior]
\label{prop-nfa-upint}
There exists a family of languages $(L_n)_{n=7,8,\ldots}$ with $\nNFA(L_n)\leq n$ and
$\nNFA(\upint L_n)\geq
2^{2^{\left\lfloor \frac{n-4}{3}\right\rfloor}}+1$.
\end{proposition}
\begin{proof}
Fix $n\geq 7$ and let $\ell={\left\lfloor \frac{n-4}{3}\right\rfloor}$.
We use two subalphabets:
$\Gamma\eqdef\{0,1,\ldots,2^\ell-1\}$ and
$\Upsilon\eqdef\{1,\ldots,\ell\}$, letting $\Sigma\eqdef\Gamma\cup\Upsilon$. The symbols in $\Gamma$, denoted
$x,y,\ldots$ are disjoint from the symbols in $\Upsilon$, denoted
$k,k',\ldots$ (for example, we can imagine that they have
different colors) and one has $\size{\Sigma}=2^\ell+\ell$.

For $x,y\in\Gamma$ and $k\in\Upsilon$, we write $x=_k y$ when $x$ and $y$,
viewed as $\ell$-bit sequences,
have the same $k$th bit. We consider the following languages:
\begin{align*}
L'_n &\eqdef\{x \, w \, y \, k \, w'\in \Gamma\cdot\Sigma^*\cdot\Gamma\cdot\Upsilon\cdot\Sigma^*~|~
  x=_{k} y\}
\:,
\\
L''_n &\eqdef \Gamma\cdot(\Gamma \cdot \Upsilon)^*
\:,
\\
L_n &\eqdef L'_n \cup (\Sigma^*\setminus L''_n)
\:.
\end{align*}
In other words, $L'_n$ contains all words such that the initial
letter $x \in \Gamma$ has one common bit with a later $y\in\Gamma$ and this
bit is indicated by the $k\in\Upsilon$ that immediately follows the
occurrence of $y$.
\begin{claim}
%\label{claim-NFA-Ln}
$\nNFA(L'_n)\leq 3\ell+2$ and
$\nNFA(L_n)\leq n$.
\end{claim}
Figure~\ref{fig-A4} displays the schematics of an NFA for $L'_n$. In order
to recognize inputs of the form $x \, w \, y \, k \, w'\in
\Gamma\cdot\Sigma^*\cdot\Gamma\cdot\Upsilon\cdot\Sigma^*$ with $x=_{k} y$,
the NFA reads the first letter $x$,
nondeterministically guesses $k$, and switches to a state $r_k^+$ or $r_k^-$
depending on what the $k$th bit of $x$ is. From there it waits
nondeterministically for the appearance of a factor $y \,k$ with $x=_k y$
before accepting. This uses $3\ell +2$ states.
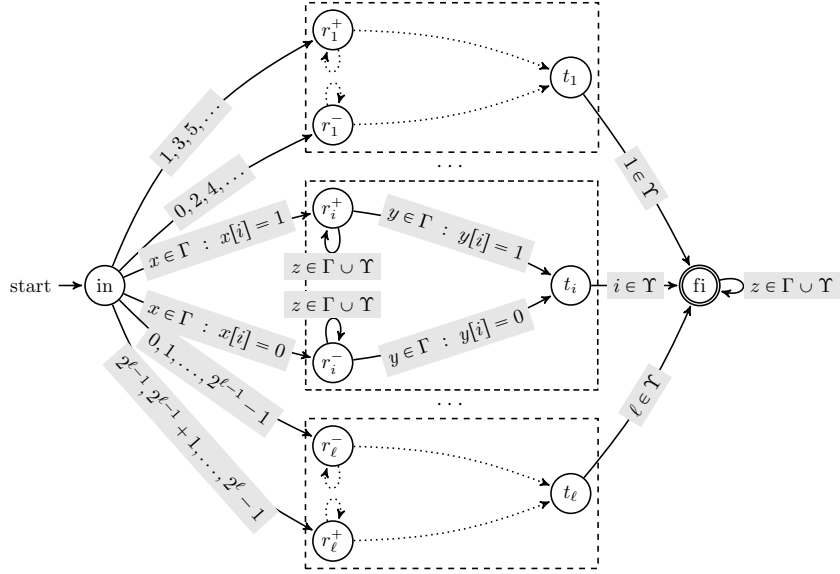
\begin{figure}[htb]
\centering
\scalebox{0.75}{
  \begin{tikzpicture}[->,>=stealth',shorten >=1pt,node distance=5em,thick,auto,bend angle=30]

\node [state] (ti) {$t_i$};
\node [above of=ti,node distance=10.5em,state] (t1) {$t_1$};
\node [below of=ti,node distance=10.5em,state] (tl) {$t_\ell$};
\node [right of=ti,node distance=6.5em,state,accepting] (fi) {f\/i};
\node [left of=fi,node distance=30em,state,initial] (in) {in};

\node [left of=t1,node distance=12em,yshift=2.4em,inner sep=0.5mm,state] (r1+) {$r_1^+$};
\node [left of=t1,node distance=12em,yshift=-2.4em,inner sep=0.5mm,state] (r1-) {$r_1^-$};
\node [left of=ti,node distance=12em,yshift=3.9em,inner sep=0.5mm,state] (ri+) {$r_i^+$};
\node [left of=ti,node distance=12em,yshift=-3.9em,inner sep=0.5mm,state] (ri-) {$r_i^-$};
\node [left of=tl,node distance=12em,yshift=2.4em,inner sep=0.5mm,state] (rl+) {$r_\ell^-$};
\node [left of=tl,node distance=12em,yshift=-2.4em,inner sep=0.5mm,state] (rl-) {$r_\ell^+$};

\node [fit=(r1-)(t1)(r1+),draw,dashed] (block1) {};
\node [fit=(ti)(ri-)(ri+),draw,densely dashed] (blocki) {};
\node [fit=(rl-)(rl+)(tl),draw,dashed] (blockl) {};

\node at ($ (block1.south)!.5!(blocki.north) $) {\large $\cdots$};
\node at ($ (blocki.south)!.5!(blockl.north) $) {\large $\cdots$};

\path[every node/.style={auto=false,sloped,fill=black!10}]
 (ri-) edge [loop above] node (e1) {$z \in \Gamma \cup \Upsilon$} (ri-)
 (ri+) edge [loop below] node (e2) {$z \in \Gamma \cup \Upsilon$} (ri+)
 (fi) edge [loop right] node [rotate=90] {$z \in \Gamma \cup \Upsilon$} (fi)
 (in) edge [bend left=20]  node {$1,3,5,\ldots$} (r1+)
 (in) edge [bend left=10]  node {$0,2,4,\ldots$} (r1-)
 (in) edge [bend left=5]   node {$x\in\Gamma~:~x[i]=1$} (ri+)
 (in) edge [bend right=5]  node {$x\in\Gamma~:~x[i]=0$} (ri-)
 (in) edge [bend right=10] node {$0,1,\ldots,2^{\ell-1}-1$} (rl+)
 (in) edge [bend right=20] node {$2^{\ell-1},2^{\ell-1}+1,\ldots,2^\ell-1$} (rl-)
 (ri+) edge [bend left=10]  node {$y\in\Gamma~:~y[i]=1$} (ti)
 (ri-) edge [bend right=10] node {$y\in\Gamma~:~y[i]=0$} (ti)
 (ti) edge node {$i \in \Upsilon$} (fi)
 (t1) edge [bend left=5] node {$1 \in \Upsilon$} (fi)
 (tl) edge [bend right=8] node {$\ell \in \Upsilon$} (fi)
;

\path[every edge/.style={draw,dotted},every loop/.style={min distance=1.5em,looseness=5.5}]
 (r1+) edge [loop below] (r1+)
 (r1+) edge [bend left=10] (t1)
 (rl+) edge [loop below] (rl+)
 (rl+) edge [bend left=10] (tl)
 (r1-) edge [loop above] (r1-)
 (r1-) edge [bend right=10] (t1)
 (rl-) edge [loop above] (rl-)
 (rl-) edge [bend right=10] (tl)
;
  \end{tikzpicture}
}%scalebox
%\vspace*{-1em}
\caption{NFA recognizing $L'_n$ with $3\ell+2$ states.}
\label{fig-A4}
%\vspace*{-1em}
\end{figure}
%%% Local Variables:
%%% fill-column: 999
%%% TeX-master: "main"
%%% End:
Adding states for $\Sigma^*\setminus L''_n$, one obtains $\nNFA(L_n)\leq
3\ell + 4\leq n$.
\\

We now consider the upward interior of $L_n$. Let
$U_\Gamma,U'_\Gamma\subseteq\Gamma^*$ be as in Lemma~\ref{lem-U-n-V}: a
word $w$ is in $U_\Gamma$ if it uses each letter from $\Gamma$ at least
once, and $U'_\Gamma \eqdef \Gamma \cdot U_\Gamma$.
\begin{claim}
%\label{lem-Ln-cap-Gamma}
$\Gamma^*\cap\upint L_n=U'_\Gamma$.
\end{claim}
  We first show $\Gamma^* \cap \upint L_n\subseteq U'_\Gamma$ by showing the
  contrapositive. Let $w \in \Gamma^* \setminus U'_\Gamma$. If $w = \emptyword$,
  then clearly $w \notin \upint L_n$. Otherwise, $w=z\,z_1\cdots z_p$, where
  $z, z_i \in \Gamma$.  Since $z_1\cdots z_p$ is not in $U_\Gamma$, there is
  some $x\in\Gamma$ that differs from all the $z_i$'s. Pick
  $k_1,\ldots,k_p$ witnessing this, that is, such that $x\neq_{k_i}z_i$ for
  all $i$. If $x=z$ we let $w'\eqdef z\, z_1 \,k_1 \cdots z_p \,k_p$
so that $w'\in L''_n$ and $w'\not\in L'_n$, that is, $w'\not\in L_n$. If $x\neq z$ we let $w'\eqdef x\,z\,k\, z_1
  \,k_1 \cdots z_p \,k_p\not\in L_n$ for some $k$ witnessing $x\neq z$, so
  that $w'\not\in L_n$. In both cases $w\subword w'\not\in L_n$ and we deduce
  $w\not\in\upint L_n$.

  We now show $U'_\Gamma \subseteq  \Gamma^*\cap\upint L_n$. Let $w = z
  \,z_1 \cdots z_p\in U'_\Gamma$. We show that $w \in \upint L_n$ by showing
  that $w' \in L_n$ for every $w'$ such that $w \subword w'$. If $w' \notin
  L''_n$, then $w' \in L_n$. So assume $w'=x\, y_1\, k_1 \cdots y_n\, k_n\in
  L''_n$. There is some $i$ such that $x=z_i$ (since $w\in U'_\Gamma$) and some $j$
  such that $z_i=y_j$ (since $w\subword w'$). We then have $x=_{k_j} y_j$
  (this does not depend on the actual value of $k_j$). Hence $w'\in L'_n
  \subseteq L_n$. Thus $w \in \upint L_n$.
\\

We are now ready to conclude the proof of Proposition~\ref{prop-nfa-upint}.
Recall that $\nNFA(L\cap\Gamma^*)\leq\nNFA(L)$ holds for any regular $L$
and any alphabet $\Gamma$. In particular the above claim entails
$\nNFA(U'_\Gamma)\leq \nNFA(\upint L_n)$. Combining with
$\nNFA(U'_\Gamma)=2^{2^\ell}+1$ from Lemma~\ref{lem-U-n-V} yields the
required $\nNFA(\upint L_n)\geq 2^{2^{\left\lfloor
    \frac{n-4}{3}\right\rfloor}}+1$.
\qed
\end{proof}

%% Local Variables:
%% ispell-check-comments: nil
%% ispell-local-dictionary: "english"
%% fill-column: 75
%% TeX-master: "main"
%% End:

% LocalWords:  DFAs th

The doubly-exponential lower bounds exhibited in
Propositions~\ref{prop-nfa-downint} and~\ref{prop-nfa-upint} rely on alphabets of
exponential size. It is an open question whether, in the case of a fixed
alphabet, the nondeterministic state complexity of downward or upward
interiors is still doubly-exponential.

\subsection{On alternating automata for closures}
%================================================
\label{ssec-afas}

The state-complexity analysis of interiors can be used to show lower bounds
on the computation of closures for regular languages represented via alternating automata (AFAs). The question was recently raised
in~\cite{holub2014} where it is suggested that the construction of a
piecewise-testable separator could be done more efficiently by using AFAs
for representing regular languages. It is indeed natural to ask whether an
AFA recognizing $\down L$ or $\up L$ can be built efficiently from an AFA
recognizing $L$, perhaps in the same spirit as the constructions for
closures on NFAs.

In the rest of this section we briefly justify the claims on AFAs made in
Table~\ref{tab-summary} in the introduction of this article. We assume
basic knowledge of AFAs (otherwise see~\cite[Section~6]{chandra76}) and
write $\nAFA(L)$ to denote the minimal number of states of an AFA
recognizing $L$. \\

For the upper bounds, recall that an AFA $A$ can be transformed into an
equivalent NFA $A'$ with the powerset construction. If $A$ has $n$ states,
$A'$ has $2^n$ states. We deduce that if $\nAFA(L)=n$, then $\nAFA(\up
L)\leq \nNFA(\up L)\leq\nNFA(L)\leq 2^n$ and $\nAFA(\down L)\leq\nNFA(\down
L)\leq\nNFA(L)\leq 2^n$. \\

For the lower bounds, we can reuse the witness languages from
Section~\ref{ssec-lower-bounds-both-interiors}. Recall the
properties of $L_n\subseteq\Sigma^*$ from
Proposition~\ref{prop-nfa-downint}. We showed that $\nNFA(L_n)\leq n$,
entailing $\nAFA(L_n)\leq n$.
Hence $\nAFA(\Sigma^*\setminus L_n)\leq n$ since one can complement an AFA
without any increase in the number of states. Let $A$ be an $\ell$-state AFA recognizing $\up
(\Sigma^*\setminus L_n)$. By complementing $A$, we get an $\ell$-state AFA
recognizing $\Sigma^*\setminus\up (\Sigma^*\setminus L_n)$, that is,  $\downint
L_n$. Transforming this into an NFA, we get a $2^\ell$-state NFA that recognizes
$\downint L_n$. Using Proposition~\ref{prop-nfa-downint} we deduce
$\ell\geq2^{\left\lfloor \frac{n-3}{2}\right\rfloor}$. Thus the languages
$(\Sigma^*\setminus L_n)_{n=3,4,\ldots}$ witness the lower bound for
$\nAFA(\up L)$ claimed in Table~\ref{tab-summary}.

For $\nAFA(\down L)$ we use the same reasoning, with up and down
interchanged, and based on the witnesses $(L_n)_{n=7,8,\ldots}$ used in
Proposition~\ref{prop-nfa-upint}.

%% Local Variables:
%% ispell-check-comments: nil
%% ispell-local-dictionary: "english"
%% fill-column: 75
%% TeX-master: "main"
%% End:

% LocalWords:  NFAs DFAs OpenFST SCCs Eq DFA NFA Kahn defn Coro AFAs AFA

\section{On unambiguous automata}
%================================
\label{sec-ufa}

Recall that an unambiguous automaton (a UFA) is an NFA $A$ in which every
accepted word is accepted by exactly one run. When handling regular
languages it is sometimes interesting to work with UFAs since, like NFAs,
they can be exponentially more succinct than DFAs and, like DFAs, they
admit polynomial-time algorithms for testing inclusion or equality,
see~\cite{colcombet2015} and references therein. With
this in mind, it was natural to state in Section~\ref{sec-closure} that
upward or downward closures are in general not more succinct when given
in the form of UFAs. We now prove these specific claims. \\

Lower bounds on the size of
UFAs can be shown via the following lemma:
\begin{lemma}[Fooling sets for unambiguous automata, after Schmidt]
\label{lem-fooling-UFA}
Given a regular language $L$ and a set of $m$ pairs of words
$S=\{(x_i,y_i)\}_{1\leq i\leq m}$, let $M_{L,S}$ be the $m\times m$ matrix
given by $M[i,j] = 1$ if $x_i y_j\in L$, and $M[i,j]=0$ otherwise. Let
$r=\rank(M_{L,S})$. Then any UFA for $L$ has at least $r$ \emph{witness}
states, where a witness state is any state that accepts at least one of the
$y_i$'s.
\end{lemma}
The above lemma is actually a refinement of Theorem~2 from~\cite{leung2005}
where the lower bound is given for $\nUFA(L)$: the proof
by Leung easily adapts to Lemma~\ref{lem-fooling-UFA} since
non-witness states contribute a null row in the matrix $M'$ one derives
from $M$ in~\cite{leung2005}.

We shall also use the following result by Leung:
\begin{lemma}[\cite{leung98}]
\label{lem-rank-subset-mat}
Let $X$ be an $n$-element set and consider $M_X$, the $2^n\times 2^n$
matrix with rows and columns indexed by subsets of $X$, given by $M[Y,Z]=1$
if $Y\cap Z\neq\emptyset$, $M[Y,Z]=0$ otherwise. Then $\rank(M_X)=2^n-1$.
\end{lemma}

As a first application, let us consider the language $\Sigma_n^*\setminus
U_n$ from Remark~\ref{rem-down-NFA} (see also Figure~\ref{fig-A2b}). Recall
that $\Sigma_n^*\setminus U_n$ contains all words where at least one letter
from $\Sigma_n$ does not occur.

\begin{proposition}
\label{prop-UFA-1}
For any $n>0$, $\nUFA(\Sigma^*_n\setminus U_n)=2^n-1$.
\end{proposition}
\begin{proof}
The upper bound is clear since already $\nDFA(\Sigma^*_n\setminus U_n)=
2^n-1$.

For the lower bound, consider $S =
\{(x_{\neg\Gamma},x_{\neg\Gamma})\}_{\Gamma\subseteq\Sigma_n}$
(recall that the words $x_\Gamma$ and $x_{\neg\Gamma}$ with $\Gamma
\subseteq \Sigma$ were introduced in the proof of Lemma~\ref{lem-U-n-V}).
The associated matrix
has $M_{\Sigma_n^*\setminus U_n,S}[x_{\neg\Gamma_i},x_{\neg\Gamma_j}]=1$ if
$x_{\neg\Gamma_i}x_{\neg\Gamma_j}\not\in U_n$, that is, if $\Gamma_i\cap\Gamma_j \neq
\emptyset$. Note that this is exactly the $M_X$ matrix from
Lemma~\ref{lem-rank-subset-mat}, instantiated with $X=\Sigma_n$. So
$\rank(M_{\Sigma_n^*\setminus U_n,S})=2^n-1$ and, by
Lemma~\ref{lem-fooling-UFA}, we can conclude that $\nUFA(\Sigma^*_n\setminus U_n) \geq 2^n-1$.
\qed
\end{proof}

The language $\down D_n$ from Section~\ref{sec-closure} is a small variation:
recall that $\down D_n$ contains all words $x\in\Sigma_n^*$
whose first suffix $x[2..]$ does not use all letters.
\begin{proposition}
\label{prop-UFA-2}
For any $n>0$,  $\nUFA(\down D_n)=2^n$.
\end{proposition}
\begin{proof}
The upper bound is clear since already $\nDFA(\down D_n)=2^n$.

For the lower bound, consider
$S=\{(a_1x_{\neg\Gamma},x_{\neg\Gamma})\}_{\Gamma\subseteq\Sigma_n} \cup
\{(\emptyword,x_{\neg\emptyset})\}$ where $a_1$ is the first letter of
$\Sigma_n$. The associated matrix $M_{\down D_n,S}$ has
\begin{align*}
 M_{\down D_n,S}[a_1x_{\neg\Gamma_i},x_{\neg\Gamma_j}]&=1 \text{  if and only if  }
\Gamma_i\cap\Gamma_j\neq\emptyset
\:,\\
M_{\down D_n,S}[\emptyword,x_{\neg\emptyset}]&=1 \:,\quad
M_{\down D_n,S}[a_1x_{\neg\Gamma_i},x_{\neg\emptyset}]=0 \:,
\\
\shortintertext{that is,}
M_{\down D_n,S} &= \left(
    \begin{array}{cc}
    \mbox{\LARGE $M_{\Sigma_n}$} & \begin{matrix} 0 \\[-0.4em] \vdots \\[-0.2em] 0 \end{matrix} \\
    \begin{matrix} 1\! & \!\!\cdots\!\! & \!1\! \end{matrix}\!\! & 1
    \end{array}
    \right)\:.
\end{align*}
We note that the column representing $x_{\neg\emptyset}$ (that is, the last column)
occurs twice in $M_{\down D_n,S}$, as the word $x_{\neg\emptyset}$ occurs twice
as the second component in $S$. One has $\rank(M_{\down D_n,S}) =
\rank(M_{\Sigma_n})+1=2^n$ since the rightmost column forbids combining
the last row with any of the earlier rows. We deduce $\nUFA(\down D_n)\geq 2^n$
with Lemma~\ref{lem-fooling-UFA}.
\qed
\end{proof}

We finally consider $\up E_{n}$ from Section~\ref{sec-closure}. Recall that $\up E_n$ contains all words over $\Sigma_n$ in which
at least one letter reappears.
\begin{proposition}
\label{prop-UFA-3}
For any $n>0$, $\nUFA(\up E_n)=2^n+1$.
\end{proposition}
\begin{proof}
The upper bound is clear since already $\nDFA(\up E_n)=2^n+1$.

For the lower bound, consider
$S=\{(x_\Gamma,x_\Gamma)\}_{\Gamma\subseteq\Sigma_n}\cup\{(a_1a_1,\emptyword)\}$.
The associated matrix $M_{\up E_n,S}$ has
\begin{align*}
M_{\up E_n,S}[x_{\Gamma_i},x_{\Gamma_j}]&=1 \text{ if and only if }
\Gamma_i\cap\Gamma_j\neq\emptyset\:,
\\
M_{\up E_n,S}[a_1a_1,x_{\Gamma_j}]&=1 \:,\quad
M_{\up E_n,S}[x_{\Gamma_i},\emptyword]=0 \:,
\\
\shortintertext{that is,}
M_{\up E_n,S} &= \left(
    \begin{array}{cc}
    \mbox{\LARGE $M_{\Sigma_n}$} & \begin{matrix} 0 \\[-0.4em] \vdots \\[-0.2em] 0 \end{matrix} \\
    \begin{matrix} 1\! & \!\!\cdots\!\! & \!1\! \end{matrix}\!\! & 1
    \end{array}
    \right)\:.
\end{align*}
Again one has $\rank(M_{\up E_n,S})=2^n$ so that, by
Lemma~\ref{lem-fooling-UFA}, any UFA for $\up E_n$ has at least $2^n$
witness states. Note however that the pairs $(x_i,y_i)$ in $S$ are such
that no $y_i$ belongs to $\up E_n$. Hence in any automaton accepting $\up
E_n$ the initial state is not a witness state. We conclude that $\nUFA(\up
E_n)\geq 2^n+1$.
\qed
\end{proof}

%% Local Variables:
%% ispell-check-comments: nil
%% ispell-local-dictionary: "english"
%% fill-column: 75
%% TeX-master: "main"
%% End:

%%  LocalWords:  UFAs Leung UBA

\section{Complexity of decision problems on closures}
%====================================================
\label{sec-decision-pbs}

In automata-based procedures for logic and verification, the state
complexity of automata constructions is not always the best measure of
computational complexity. In this section we gather some elementary
results on the complexity of subword-related decision problems for
automata: for finite automata $A$, $B$ we want to know whether
the accepted language $L(A)$ is downward
or upward closed, respectively, and  whether $L(A)$ and $L(B)$ have the same
downward or upward closures. These questions are in the
spirit of the work done
in~\cite{brzozowski2011b,kao2009,rampersad2012} for various notions of
closures.
Some of the results we give
are already known but are scattered in the literature and
sometimes even reappear as open questions (see, for example,~\cite{fu2011}).

\subsection{Deciding closedness}
%===============================

Deciding whether $L(A)$ is upward-closed
or downward-closed is, unsurprisingly, $\PSPACE$-complete for NFAs, and
$\NL$-complete for DFAs. For upward-closedness this is already shown
in~\cite{heam2002}, and quadratic-time algorithms that decide
upward-closedness of $L(A)$ for a DFA $A$ already appear
in~\cite{arfi87,pin97}.

\begin{proposition}
Deciding whether $L(A)$ is upward-closed or downward-closed is
$\PSPACE$-complete when $A$ is an NFA, even in the
2-letter alphabet case.
\end{proposition}
\begin{proof}
Membership in $\PSPACE$ is clear since it is enough to decide whether
$A$ and $A^\up$ or $A^\down$ accept the same language.

$\PSPACE$-hardness can be shown by adapting the
proof for hardness of universality. Let $R$ be a length-preserving
semi-Thue system and $x$, $x'$ two strings of same length. It is
$\PSPACE$-hard to say whether there is a derivation $x\stepR{*} x'$, even for a fixed
$R$ over a 2-letter alphabet $\Sigma$. We reduce (the negation of) this
question to our problem.

Fix $x$ and $x'$ of length $n>1$: a word $x_1\,x_2\cdots x_m$ of
length $n\times m$ encodes a derivation if $x_1=x$, $x_m=x'$, and
$x_i\stepR{} x_{i+1}$ for all $i=1,\ldots,m-1$. The language
$L_{R,x,x'}$ of words that do \emph{not} encode a derivation from $x$
to $x'$ is regular and recognized by an NFA with $O(n)$ states. Now,
there is a derivation $x\stepR{*} x'$ if and only if $L_{R,x,x'}\neq\Sigma^*$. We
conclude by observing that $L_{R,x,x'}=\Sigma^*$ if and only if $L_{R,x,x'}$ is
upward-closed or, equivalently, downward-closed; this is because
$L_{R,x,x'}$ contains all words of length not divisible by $n>1$.
\qed
\end{proof}
\begin{proposition}
Deciding whether $L(A)$ is upward-closed or downward-closed is
$\NL$-complete when $A$ is a DFA, even in the
2-letter alphabet case.
\end{proposition}
\begin{proof}
We only prove the result for upward-closure since
 $L$ is down\-ward-closed if and only if $\Sigma^*\setminus L$ is
upward-closed, and since one easily builds a DFA for the complement of
$L(A)$.

For membership in $\NL$,
we first observe that $L$ is upward-closed if and only if, for all
$u,v\in\Sigma^*$, $u\,v\in L$ implies $u\,a\,v\in L$ for all $a\in\Sigma$.
Therefore, $L(A)$ is not upward-closed -- for $A=(\Sigma,Q,\delta,q_\init,F)$ -- if and only if there are states $p,q\in
Q$, a letter $a$, and words $u,v$ such that $\delta(q_\init,u)=p$,
$\delta(p,a)=q$, $\delta(p,v)\in F$ and $\delta(q,v)\notin F$. If
such words exist, one can, in particular, find witnesses with  $\size{u}<n$ and
$\size{v}<n^2$ where $n=\size{Q}$ is the number of states of $A$.
Hence checking that $L(A)$ is not upward-closed can be performed in nondeterministic logarithmic space by
guessing $u$, $a$, and $v$ within the above length bounds, finding $p$
and $q$ by running $u\,a$ from $q_\init$, then running $v$ from
both $p$ and $q$. Since $\coNL=\NL$, we conclude that  
upward-closedness too is in $\NL$. 

For $\NL$-hardness, one may reduce from vacuity of DFAs, a well-known
$\NL$-hard problem that is essentially equivalent to $\GAP$, the  Graph
Accessibility Problem. Note that for any DFA, and in fact any NFA, $A$
with $n$ states the
following equivalences hold:
\begin{equation*}
L(A)\cap\Sigma^{< n}\text{ is upward-closed }
\iff
L(A)\cap \Sigma^{< n}=\emptyset
\iff
L(A)=\emptyset \:.
\end{equation*}
This provides the required reduction
since, given a DFA $A$, one easily
builds a DFA for $L(A)\cap\Sigma^{< n}$ in logspace.
\qed
\end{proof}

\subsection{Deciding equivalence modulo closure}
%===============================================
\label{ssec-equiv-mod-closure}

The question whether $\down L(A)=\down L(B)$ or, similarly, whether
$\up L(A)=\up L(B)$, is relevant in some settings where closures are
used to build regular over-approximations of more complex languages.

Bachmeier \textit{et al}.\  recently showed that the above two
questions are $\coNP$-complete when $A$ and $B$ are
NFAs~\cite[Section~5]{bachmeier2015}, hence ``easier'' than deciding
whether $L(A)=L(B)$. Here we give an improved version of their
result.

\begin{proposition}[after~\cite{bachmeier2015}]
1.\  Deciding whether $\down L(A)\subseteq \down L(B)$ or whether $\up
  L(A)\subseteq \up L(B)$ is $\coNP$-complete when $A$ and $B$ are
  NFAs.

2.\ Deciding $\down L(A)=\down L(B)$ or
  $\up L(A)=\up L(B)$ is $\coNP$-hard even when $A$ and $B$
are DFAs over a two-letter alphabet.

3.\ These problems are $\NL$-complete when restricting to NFAs over a
1-letter alphabet.
\end{proposition}
\begin{proof}
1.\  Let $B=(\Sigma,Q,\delta,I,F)$ and $n_B=\size{Q}$. Assume that $\down
  L(A)\not\subseteq\down L(B)$ and pick a shortest witness
  $x=x_1\cdots x_\ell\in\Sigma^*$ with $x\in \down L(A)$ and
  $x\not\in\down L(B)$. We claim that $\size{x}<n_B$: indeed in the
  deterministic powerset automaton obtained from $B^\down$, the unique run
  $S_0\step{x_1}S_1\step{x_2}\cdots \step{x_\ell}S_\ell$ of $x$ is
  such that $Q=S_0\supseteq S_1\supseteq S_2\cdots \supseteq S_\ell\neq\emptyset$
  (recall the proof of Lemma~\ref{lem-dc-alphabet-size}).
  If $S_{i-1}=S_i$ for some $i$, a shorter witness is obtained by
  omitting the $i$th letter in $x$: this does not affect membership in
  $\down L(A)$ since this language is downward-closed. One concludes
  that the $S_i$ have strictly diminishing sizes, hence $\ell<n_B$.
  This leads to an $\NP$ algorithm deciding $\down L(A)\not\subseteq\down
  L(B)$: guess $x$ in $\Sigma^{<n_B}$ and check in polynomial time
  that it is accepted by $A^\down$ and not by $B^\down$.

  For upward closure the reasoning is even simpler: a shortest witness
  $x$ with $x\in\up L(A)$ and $x\not\in\up L(B)$ has length
  $\size{x}<n_A$: if $x$ is longer, a pumping lemma allows
  one to find a subword $x'\in\up L(A)$, and $x'\not\in\up L(B)$ since
  $x\not\in\up L(B)$.
  \\

2.\ $\coNP$-hardness is shown by reduction from validity of
DNF-formulae. Consider an arbitrary DNF formula $\phi=C_1\lor C_2\lor
\cdots \lor C_m$ consisting of $m$ conjunctions of literals where $k$ Boolean
variables $v_1,\ldots,v_k$ may appear, for example, $\phi = (v_1\land \neg v_2\land
v_4)\lor (v_2\land\cdots)\lor\cdots$. The language of all the
valuations, seen as words in $\{0,1\}^k$, under which $\phi$ holds true
is recognized by an NFA that has size $O(\size{\phi}^2)$.
We slightly modify this language so that we can use a DFA
instead of an NFA.  Let
$L_\phi=\{1^\ell 0 \, x_1 \cdots x_k\in\{0,1\}^*~|~ 0\leq\ell<m\land x_1\cdots x_k\models
C_{\ell+1}\}$. We build a DFA $A_\phi$, having $m(k+2)$ states, that recognizes $L_\phi$: see
Figure~\ref{fig-red-taut} where, for the sake of readability, the
picture uses wavy edges where $A_\phi$ recognizes a $1^\ell 0$ prefix,
and standard edges where it recognizes the encoding of a valuation $x_1\cdots x_k$ proper.
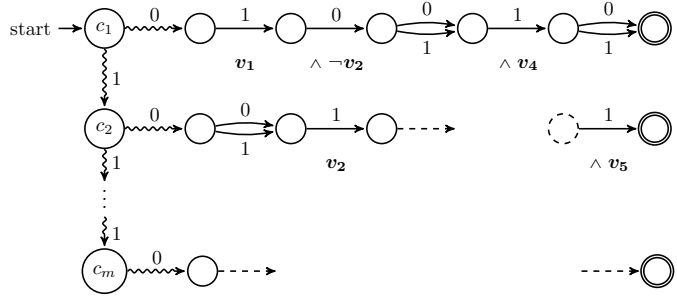
\begin{figure}[htbp]
\centering
\scalebox{0.75}{
  \begin{tikzpicture}[->,>=stealth',shorten >=1pt,node distance=6em,thick,auto,bend angle=30]
\tikzstyle{every state}=[minimum size=1.5em]
\node [state,initial] (c1) {$c_1$};
\node [state,below=3em of c1] (c2) {$c_2$};
\node [state,below=5em of c2] (cm) {$c_m$};
\node at ($(c2)!0.5!(cm)$) [minimum size=0em,inner sep=0.4em] (cdots) {$\rvdots$};

\node [state,right=3em of c1] (c11) {};
\node [state,right=3em of c11] (c12) {};
\node [state,right=3em of c12] (c13) {};
\node [state,right=3em of c13] (c14) {};
\node [state,right=3em of c14] (c15) {};
\node [state,right=3em of c15,accepting] (c16) {};

\node [state,right=3em of c2] (c21) {};
\node [state,right=3em of c21] (c22) {};
\node [state,right=3em of c22] (c23) {};
\node [state,draw=white,right=3em of c23] (c24) {};
\node [state,dashed,right=3em of c24] (c25) {};
\node [state,right=3em of c25,accepting] (c26) {};

\node [state,right=3em of cm] (cm1) {};
\node [state,draw=white,right=3em of cm1] (cm2) {};
\node [state,draw=white,right=3em of cm2] (cm3) {};
\node [state,draw=white,right=3em of cm3] (cm4) {};
\node [state,draw=white,right=3em of cm4] (cm5) {};
\node [state,right=3em of cm5,accepting] (cm6) {};

\path (c11) edge node {$1$} node [swap,outer sep=1.2em] {$\bm{v_1}$} (c12);
\path (c12) edge node {$0$} node [swap,outer sep=1.2em] {$\bm{\land\:\neg v_2}$} (c13);
\path (c13) edge [bend left=10] node {$0$} (c14);
\path (c13) edge [bend right=10] node [swap] {$1$} (c14);
\path (c14) edge node {$1$} node [swap,outer sep=1.2em] {$\bm{\land\: v_4}$} (c15);
\path (c15) edge [bend left=10] node {$0$} (c16);
\path (c15) edge [bend right=10] node [swap] {$1$} (c16);

\path (c21) edge [bend left=10] node {$0$} (c22);
\path (c21) edge [bend right=10] node [swap] {$1$} (c22);
\path (c22) edge node {$1$}  node [swap,outer sep=1.2em] {$\bm{v_2}$} (c23);
\path (c25) edge node {$1$}  node [swap,outer sep=1.2em] {$\bm{\land\:v_5}$} (c26);

\path (c23) edge [dashed] (c24);
\path (cm1) edge [dashed] (cm2);
\path (cm5) edge [dashed] (cm6);

{\tikzstyle{every edge}=[draw,decorate,decoration={snake,amplitude=.3mm,segment length=1.5mm,post length=2mm}]
\path (c1) edge [bend left=0] node {$1$} (c2);
\path (c2) edge [bend left=0,pos=0.4] node {$1$} (cdots);
\path (cdots) edge [bend left=0,pos=0.54] node {$1$} (cm);
%\path (cm) edge [bend left=24] node {$1$} (c1);
\path (c1) edge node {$0$} (c11);
\path (c2) edge node {$0$} (c21);
\path (cm) edge node {$0$} (cm1);
}

  \end{tikzpicture}
}%scalebox
%\vspace*{-1em}
\caption{DFA $A_\phi$ for $\phi = (v_1\land \neg v_2\land v_4)\lor (v_2\land\cdots
  \land v_5)\lor\cdots \lor C_m$ with $k=5$ variables.}
\label{fig-red-taut}
%\vspace*{-1em}
\end{figure}
%%% Local Variables:
%%% fill-column: 999
%%% End:

Now let $B_\phi$ be a
DFA for $L_\phi\cup 1^m 0(0+1)^k$, where all valuations
are allowed after the $1^m0$ prefix, and observe that $\up
L(A_\phi)=\up L(B_\phi)$ if and only if $1^m 0(0+1)^k\subseteq\up
L(A_\phi)$. However, $1^m 0\,x_1\cdots x_k\in \up L(A_\phi)$ requires
that $1^\ell 0\,x_1\cdots x_k\in  L(A_\phi)$ for some $\ell\leq m$.
Finally, $\up L(A_\phi)=\up L(B_\phi)$ if and only if
 all valuations make $\phi$ true, that is, if
$\phi$ is valid. Since $A_\phi$ and $B_\phi$ are built in logspace
from $\phi$, this completes the reduction for equality of upward
closures.

For downward closures, we modify $A_\phi$ by adding a transition
$c_m\step{1}c_1$ so that the resulting $A'_\phi$ accepts all words $1^\ell 0\,
x_1 \cdots x_k$ such that $x_1\cdots x_k$ makes $C_{\ell'+1}$
true for $\ell'=\ell\mod m$. For $B$ we now take a DFA for $1^*0(0+1)^k$ and see that $\down
L(A'_\phi)=\down L(B)$ if and only if all valuations make $\phi$ true.
\\

3.\  In the 1-letter case, comparing upward or downward closures
amounts to comparing the length of the shortest or longest word, respectively,
accepted by the automata. This is easily done in nondeterministic
logspace. And since $\up L(A)=\down L(A)=\emptyset$ if and only if
$L(A)=\emptyset$, $\NL$-hardness is shown by reduction from emptiness
of NFAs, that is, a question ``is there a path from an initial state to
an accepting state'' which is
just another version of $\GAP$, the Graph Accessibility Problem.
\qed
\end{proof}

A special case of language comparison is testing for universality. The question whether
$\up L(A)=\Sigma^*$ is trivial since it amounts to asking whether
$\emptyword$ is accepted by $A$. For downward closures one has the
following:
\begin{proposition}[after~\cite{rampersad2012}]
Deciding whether $\down L(A)=\Sigma^*$ when $A$ is an NFA over $\Sigma$
is $\NL$-complete.
\end{proposition}
\begin{proof}
Rampersad \textit{et al}.\  show that the problem can be solved in
linear time~\cite[Section~4.4]{rampersad2012}. Actually the
characterization they use, namely
``$\down L(A)=\Sigma^*$ if and only if $A=(\Sigma,Q,\delta,I,F)$ has a
state $q\in Q$ with $I\step{*}q\step{*}F$ and such that for any
$a\in\Sigma$ there is a path of the form $q\step{*}\step{a}\step{*}q$
from $q$ to itself'',
is a $\FO+\TC$ sentence on $A$ seen as a directed labeled graph, hence can be
checked in $\NL$~\cite{immerman87b}.
$\NL$-hardness can be shown by reduction from
emptiness of NFAs, for example, by adding loops $p\step{a}p$ on any accepting
state $p\in F$ and for every $a\in\Sigma$.
\qed
\end{proof}

%% Local Variables:
%% ispell-check-comments: nil
%% ispell-local-dictionary: "english"
%% fill-column: 70
%% TeX-master: "main"
%% End:

% LocalWords:  SCC SCC's NFAs DFAs DFA NFA FSA Bachmeier et al th DNF
%%  LocalWords:  Rampersad

\section{Concluding remarks}
%===========================
\label{sec-concl}

We considered the state complexity of ``closure languages'' obtained by
starting from an arbitrary regular language and closing it with all its
subwords or all its superwords. These closure operations are
essential when reasoning with subwords~\cite{KS-fosubw}. We completed
the known results on closures by providing exact state complexities in the
case of unbounded alphabets, and by demonstrating an exponential lower
bound on downward closures even in the case of a two-letter alphabet.

We also considered the dual notion of computing interiors.
The nondeterministic state complexity of interiors is a new problem that we
introduced in this article and for which we show
doubly-exponential upper and lower bounds. From this we can deduce
an exponential state complexity for the upward and downward closures of
languages represented via alternating automata.
\\

These results contribute to a more general research agenda: what are the ``best''
data structures and algorithms for reasoning with subwords and
superwords? The algorithmics of subwords and superwords has mainly been
developed in string matching and
combinatorics~\cite{baezayates91,elzinga2008}. When considering subwords
and superwords for sets of strings rather than individual strings -- in
matching and combinatorics~\cite{tronicek2005} but also in other fields
like model-checking and constraint solving~\cite{hooimeijer2011,KS-fosubw} --, there
are many different ways of representing downward-closed and upward-closed
sets. Automata-based representation are not always the preferred option;
see, for example, the SREs used for downward-closed languages
in~\cite{abdulla-forward-lcs}. The existing trade-offs between all the
available options are not yet well understood and certainly deserve more
scrutiny.
In this direction, let us mention~\cite[Theorem~2.1(3)]{birget93} showing
that if $\nDFA(L)=n$ then $\min(L)\eqdef\{x\in L ~|~ \forall y\in L :
y\subword x\implies y=x\}=L\setminus(L\shuffle \Sigma)$ may have
$\nNFA(\min(L))=(n-2)2^{n-3}+2$, to be contrasted with $\nNFA(\up L)\leq n$. This suggests that it is more efficient to
represent $\up L$ directly than by its minimal elements.

\subsubsection*{Acknowledgments.} We thank S.\ Schmitz and the anonymous
reviewers for their many comments and suggestions that helped improve the final
version of this article.

%% Local Variables:
%% ispell-check-comments: nil
%% ispell-local-dictionary: "english"
%% fill-column: 75
%% End:

% LocalWords:  NFAs SREs AFAs

\section*{References}
\bibliographystyle{plain}
\bibliography{subwords}

\end{document}